\documentclass[12pt,draftcls,onecolumn]{IEEEtran}
\usepackage{comment,url}
\newcommand{\R}{{\cal R}}
\newcommand{\V}{{\cal V}}
\newcommand{\A}{{\cal A}}
\newcommand{\G}{{\cal G}}
\newcommand{\X}{{\cal X}}
\newcommand{\D}{{\cal D}}
\newcommand{\W}{{\cal W}}
\newcommand{\BB}{{\cal B}}
\newcommand{\E}{{\cal E}}
\newcommand{\T}{{\cal T}}

\renewcommand{\H}{{\cal H}}

\newcommand{\diag}{{\rm diag}}

\usepackage{graphicx,amsmath,amsbsy,amssymb}
\usepackage{colortbl}
\usepackage{subfigure}
\usepackage{tabularx}
\usepackage{theorem}
\usepackage{algorithm,algorithmic}
\theorembodyfont{\upshape}
\newtheorem{theorem}{Theorem}
\newtheorem{assumption}{Assumption}
\newtheorem{lemma}{Lemma}
\newtheorem{proposition}{Proposition}

\newtheorem{definition}{Definition}

\newtheorem{proof}{Proof}

\theoremstyle{break}

\makeatletter
\newcommand{\tblcaption}[1]{\def\@captype{table}\caption{#1}}  
\makeatletter

\begin{document}
\title{
Payoff-based Inhomogeneous Partially Irrational Play for Potential Game
Theoretic Cooperative Control of Multi-agent Systems}
\author{Tatsuhiko Goto, Takeshi Hatanaka,~\IEEEmembership{Member,~IEEE} and
Masayuki Fujita,~\IEEEmembership{Member,~IEEE} 
\thanks{Tatsuhiko Goto is with Toshiba Corporation,
Takeshi Hatanaka(corresponding author)
and Masayuki Fujita
are with the Department of Mechanical and Control Engineering,
Tokyo Institute of Technology, Tokyo 152-8550,
JAPAN, {\sf hatanaka@ctrl.titech.ac.jp}, {\sf fujita@ctrl.titech.ac.jp}}
}

%------------------------------------------------------------------------%
% MakeTitle
%------------------------------------------------------------------------%
\maketitle
\begin{abstract}
This paper handles a kind of strategic game called potential games
and develops a novel learning algorithm 
Payoff-based Inhomogeneous Partially Irrational Play (PIPIP).
The present algorithm is based on Distributed Inhomogeneous 
Synchronous Learning (DISL) presented in an existing work but, unlike DISL,
PIPIP allows agents to make irrational decisions with a specified probability, 
i.e. agents can choose an action with a low utility from the past actions 
stored in the memory.
Due to the irrational decisions,
we can prove convergence in probability
of collective actions to potential function maximizers. 
Finally, we demonstrate the effectiveness of the present algorithm
through experiments on a sensor coverage problem.
It is revealed through the demonstration that
the present learning algorithm successfully leads
agents to around potential function maximizers
even in the presence of undesirable Nash equilibria.
We also see through the experiment
with a moving density function that PIPIP has adaptability
to environmental changes.
\end{abstract}
\IEEEpeerreviewmaketitle

\section {Introduction}
\label{sec:1}

Cooperative control of multi-agent systems
basically aims at designing local interactions
of agents in order to meet some global objective of the group
\cite{BCM09,M07}.
It is also required depending on scenarios
that agents achieve the global objective under
imperfect prior knowledge on environments
while adapting to the network and environmental changes.
Nevertheless, conventional cooperative control schemes
do not always embody such functions.
For example, in sensor deployment or coverage,
most of the control schemes as in \cite{CMKB04,LC05,CZ08}
assume prior knowledge on
a density function defined over a mission space 
and hence are hardly applicable to the mission over unknown surroundings.
A game theoretic framework as in \cite{MAS09} 
holds tremendous potential for overcoming the drawback of 
the conventional schemes.

A game theoretic approach to cooperative control 
formulates the problems as non-cooperative games
and identifies the objective in cooperative control
with arrival at some specific Nash equilibria \cite{MAS09,ZM09a,LM10}.
In particular, it is shown by J. Marden et al. \cite{MAS09} that
a variety of cooperative control problems are related to
so-called potential games \cite{MS96}. 
Unlike the other game theory,
potential games give a design perspective,
which consists of two kinds of design problem: 
utility design and learning algorithm design \cite{wierman}.
The objective of utility design is to align
local utility functions to be maximized by each agent
so that the resulting game constitutes a potential game,
where the literature \cite{S_53,WT_99} provides 
general design methodologies.
The learning algorithm design determines action selection
rules of agents so that the actions converge to Nash equilibria.

In this paper, we focus on the learning algorithm design
for cooperative control of multi-agent systems.
A lot of learning algorithms have been established 
in game theory literature and recently some algorithms
are also developed mainly by J. Marden and his collaborators.
The algorithms therein are classified
into several categories depending on their features.

The first issue is whether an algorithm presumes finite or
infinite memories.
For example, Fictitious Play (FP) \cite{MS96b}, Regret Matching (RM) \cite{SA03}, 
Joint Strategy Fictitious Play (JSFP) with Inertia  \cite{JGS09} and 
Regret-Based Dynamics \cite{JGS07} require infinite number of 
memories for executing the algorithms.
Meanwhile, Adaptive Play (AP) \cite{Y93}, Better Reply Process with Finite Memory and Inertia
\cite{Y04}, (Restrictive) Spatial Adaptive Play ((R)SAP) \cite{Young,MAS09} and Payoff-based 
Dynamics (PD) \cite{JYGS09}, Payoff-based version of Log-Linear 
Learning (PLLL) \cite{MS08} and 
Distributed Inhomogeneous Synchronous 
Learning (DISL) \cite{ZM09a}
require only a finite number of memories.
Of course, the finite memory algorithms are more preferable
for practical applications.

The second issue is what information is necessary
for executing learning algorithms.
For example, FP presumes that all the information of the
other agents' actions are available, which strongly restricts its applications.
On the other hand, RM, JSFP with Inertia and (R)SAP 
assume availability of a so-called virtual payoff, i.e.
the utility which would be obtained if an agent chose an action.
%It is a milder assumption than the perfect information as in FP.
%but still limits their applications depending on application scenarios.
%For example, in sensor coverage over unknown environments,
%each sensor never gains a reward before executing actions.
%In order to overcome such a drawback, some algorithms such as
Moreover, PD, PLLL and DISL 
utilize only the actual payoffs obtained after taking actions,
which has a potential to overcome the aforementioned
drawback of the sensor coverage schemes \cite{ZM09a}.

The main objective of standard game theory is to
compute Nash equilibria and hence 
most of the above algorithms except for \cite{MAS09,MS08}
assure only convergence to pure Nash equilibria.
However, in most of cooperative control problems,
it is insufficient for achieving the global objective
and selection of the most efficient equilibria is required \cite{MS08}.
In this paper, we thus deal with convergence of the actions to the
Nash equilibria maximizing the potential function which are called
optimal Nash equilibria in this paper,
since the potential function is usually designed in many
cooperative control problems so that its maximizers
coincide with the action profiles achieving the global objectives.

The primary contribution of this paper is to develop a novel learning algorithm
called {\it Payoff-based Inhomogeneous Partially Irrational Play} (PIPIP).
The learning algorithm is based on DISL presented in \cite{ZM09a} and
inherits its several desirable features: 
(i) The algorithm requires finite and a little memory,
(ii) The algorithm is payoff-based,
(iii) The algorithm allows agents to choose actions 
in a synchronous fashion at each iteration,
(iv) The action selection procedure in PIPIP consists of simple rules,
(v) The algorithm is capable of dealing with
constraints on action selection.
%is payoff-based and does not assume knowledge on 
%virtual payoffs which would be obtained if an agent chose an action.
%finite memory, payoff-based, synchronous, capability dealing with
%constraints on action selection and simple rules.
The main difference of PIPIP from DISL is to allow agents to make irrational 
decisions with a certain probability, which
renders agents opportunities to escape from undesirable 
Nash equilibria.
Thanks to the irrational decisions, PIPIP assures that the actions
of the group converge in probability to optimal Nash equilibria,
though only convergence to a pure Nash equilibrium is proved
in \cite{ZM09a}.
Meanwhile, some learning algorithms as in \cite{MAS09,MS08}
dealing with convergence to the optimal 
Nash equilibria have been presented and
we also mention the advantages of PIPIP over these
learning algorithms in the following.
RSAP \cite{MAS09} guarantees convergence of the distribution 
of actions to a stationary distribution such that the probability
staying the optimal Nash equilibria is arbitrarily
specified by a design parameter.
However, RSAP is not synchronous and virtual payoff-based and hence
its applications are restricted.
PLLL \cite{MS08} also
allows irrational and exploration decisions similarly to PIPIP and
the resulting conclusion is almost compatible with this paper.
However, in \cite{MS08}, how to handle the action constraints is not explicitly shown
and convergence in probability to the optimal Nash equilibria
is not proved in a strict sense.
%Furthermore, the probabilities 
%taking irrational and exploration decisions are strongly coupled
%in \cite{MS08}, which may cause difficulty in practical applications.

The secondary contribution of this paper is to
demonstrate the effectiveness of the present 
learning algorithm through experiments on a sensor coverage problem,
where the learning algorithm is applied to a robotic system
compensated by local controllers and logics.
Such investigations have not been sufficiently addressed
in the existing works.
Here, we mainly check the performance of the learning algorithm 
in finite time and adaptability to environmental changes.
In order to deal with the former issue, we prepare obstacles in the mission space
to generate apparent undesirable Nash equilibria.
Then, we compare the performance of PIPIP with DISL.
The results therein will support our claim that what this paper provides
is not a minor extension of \cite{ZM09a} and contains a significant 
contribution from a practical point of view.
We next demonstrate the adaptability
by employing a moving density function defined over the mission space.
Though adaptation to time-varying density is in principle expected 
for payoff-based algorithms, its demonstration has not been addressed 
in previous works.
We see from the results that desirable 
group behaviors, i.e. tracking to the moving high density region
are achieved by PIPIP even in the absence of any knowledge on the density.

This paper is organized as follows:
In Section \ref{sec:2}, we give some terminologies and basis
necessary for stating the results of this paper.
In Section \ref{sec:3}, 
we present the learning algorithm PIPIP
and state the main result associated with the algorithm,
i.e. convergence in probability to the optimal Nash equilibria.
Then, Section \ref{sec:4} gives the proof of the main result.
In Section \ref{sec:5}, we demonstrate the effectiveness of PIPIP
through experiments on a sensor coverage problem.
Finally, Section \ref{sec:6} draws conclusions.

\section{Preliminary}
\label{sec:2}
\subsection{Constrained Potential Games}
\label{sec:2.1}

In this paper, we consider a constrained strategic game
$\Gamma = (\V, \A, \{U_i(\cdot)\}_{i\in \V},\{\R_i(\cdot)\}_{i\in \V})$.
Here, ${\mathcal V} :=\{1,\cdots,n \}$ is the set of agents' 
unique identifiers. 
The set $\A$ is called a collective action set and defined as
$\A := \A_1 \times \cdots \times \A_n$, where
$\A_i,\ i\in \V$ is the set of actions which agent $i$ can take.
The function $U_i: \A \rightarrow {\mathbb R}$ is a so-called utility function
of agent $i\in \V$ and each agent basically 
behaves so as to maximize the function.
The function $\R_i: \A_i \rightarrow 2^{\A_i}$ 
provides a so-called 
constrained action set and $\R_i(a_i)$ is the set of actions which 
agent $i$ will be able to take in case he takes an action $a_i$.
Namely, at each iteration $t \in {\mathbb Z}_+ := \{0,1,2,\cdots\}$, each agent chooses
an action $a_i(t)$ from the set $\R_i(a_i(t-1))$.

Throughout this paper, we denote 
%a group action by $a = (a_1,\cdots, a_n)$ and 
collection of actions other than agent $i$ by 
\begin{equation*}
 a_{-i}:=(a_{1},\cdots,a_{i-1},a_{i+1},\cdots,a_{n}).
\end{equation*}
Then, the joint action $a =(a_{1},\cdots,a_{n}) \in \A$ is described as $a=(a_i,a_{-i})$.
Let us now make the following assumptions.
%%%%%%%%%%%%%%%%%%%%%%%%%%%%%%%%%%%%%%%%%%%%%%%%%%%%%%
% Assumption 1 
%%%%%%%%%%%%%%%%%%%%%%%%%%%%%%%%%%%%%%%%%%%%%%%%%%%%%%
\begin{assumption}
\label{ass:1}
The function $\R_i: \A_i \rightarrow 2^{\A_i}$  satisfies 
the following three conditions.
\begin{itemize}
\item (Reversibility \cite{MAS09})
For any $i \in {\mathcal V}$ and any actions $a_i^1, a^2_i \in {\mathcal 
	   A}_i$, the inclusion $a^2_i\in \R_i(a^1_i)$ is equivalent to
$a^1_i\in \R_i(a^2_i)$.
\item (Feasibility \cite{MAS09})
For any $i \in {\mathcal V}$ and any actions $a_i^1, a^m_i \in {\mathcal 
	   A}_i$, there exists a sequence of actions
$a_i^1 \rightarrow a_i^2 \rightarrow \cdots \rightarrow a^m_i$
satisfying $a_i^l \in \R_i(a_i^{l-1})$ for all $l \in \{1,\cdots,m\}$.
\item For any $i \in {\mathcal V}$ and any action $a_i \in \A_i$,
the number of available actions in $\R_i(a_i)$ is greater than or equal to $3$.
\end{itemize}
\end{assumption}
%%%%%%%%%%%%%%%%%%%%%%%%%%%%%%%%%%%%%%%%%%%%%%%%%%%%%%%

%%%%%%%%%%%%%%%%%%%%%%%%%%%%%%%%%%%%%%%%%%%%%%%%%%%%%%
% Assumption 2
%%%%%%%%%%%%%%%%%%%%%%%%%%%%%%%%%%%%%%%%%%%%%%%%%%%%%%
\begin{assumption}\label{ass:2}
For any $(a, a')$ satisfying 
$a'_i\in \R_i(a_i)$ and $a_{-i}=a'_{-i}$, 
the inequality $U_i(a')-U_i(a) < 1$ holds true
for all $i \in {\mathcal V}$.
\end{assumption}
%%%%%%%%%%%%%%%%%%%%%%%%%%%%%%%%%%%%%%%%%%%%%%%%%%%%%%%
Assumption \ref{ass:2} means that when only one agent changes his action, the 
difference in the utility function $U_i$ should be smaller than $1$.
This assumption is satisfied by just scaling 
all agents' utility functions appropriately. 

Let us now introduce the potential games 
under consideration in this paper.
%%%%%%%%%%%%%%%%%%%%%%%%%%%%%%%%%%%%%%%%%%%%%%%%%%%%%%
% Definition 1
%%%%%%%%%%%%%%%%%%%%%%%%%%%%%%%%%%%%%%%%%%%%%%%%%%%%%%
\begin{definition}[Constrained Potential Games \cite{MAS09,ZM09a}] \label{def:1}
A constrained strategic game $\Gamma$ is said to 
be a constrained potential game with 
potential function $\phi :{\cal A}\rightarrow {\mathbb R}$
if for all $i\in {\mathcal V}$,
every $a_{i}\in {\cal A}_i$
and every $a_{-i} \in \prod_{j\neq i}\A_j$,
the following equation holds 
for every $a'_i\in \R_i(a_i)$.
\begin{equation}
U_i(a'_i,a_{-i})-U_i(a_i,a_{-i})=\phi (a'_i,a_{-i})-\phi (a_i,a_{-i})
\label{eq:2.1}
\end{equation}
\end{definition}
%%%%%%%%%%%%%%%%%%%%%%%%%%%%%%%%%%%%%%%%%%%%%%%%%%%%%%%

Throughout this paper, we suppose that a potential function $\phi$ 
is designed so that its  maximizers coincide with
the joint action $a$ achieving a global objective of the group.
Under the situation, (\ref{eq:2.1})
implies that if an agent changes his action,
the change of the local objective function is equal to
that of the group objective function.

We next define the Nash equilibria as below.
%%%%%%%%%%%%%%%%%%%%%%%%%%%%%%%%%%%%%%%%%%%%%%%%%%%%%%
% Definition 2
%%%%%%%%%%%%%%%%%%%%%%%%%%%%%%%%%%%%%%%%%%%%%%%%%%%%%%
\begin{definition}[Constrained Nash Equillibria]
\label{def:2}
For a constrained strategic game $\Gamma$,
a collection of actions $a^*\in \A$ is said to be a constrained
pure Nash equilibrium if
the following equation holds for all $i \in {\mathcal V}$.
\begin{equation}
U_i(a^*_i,a^*_{-i})=\max_{a_i\in \R_i(a_i^*)}U_i(a_i,a^*_{-i})
\label{eq:2.2}
\end{equation}
\end{definition}
%%%%%%%%%%%%%%%%%%%%%%%%%%%%%%%%%%%%%%%%%%%%%%%%%%%%%%
It is known \cite{ZM09a,MS96} that any constrained potential game
has at least one pure Nash equilibrium and, in particular,
a potential function maximizer
is a Nash equilibrium, which is called
an optimal Nash equilibrium in this paper.
However, there may exist undesirable pure
Nash equilibria not maximizing the potential function.
%, but its inverse is not always true.
%Thus, a strategy that the agents just try to
%maximize its own utility does not lead the group 
%to the maximal potential function.
In order to reach the optimal Nash equilibria
while avoiding undesirable equilibria, 
we have to design appropriately
a learning algorithm which 
determines how to select an action at each iteration.

\subsection{Resistance Tree}
\label{sec:2.2}

Let us consider a Markov process $\{P_t^0\}$ defined over a finite state space $\X$.
A perturbation of $\{P_t^0\}$ is a Markov 
process whose transition probabilities are slightly perturbed.
Specifically, a perturbed Markov process 
$\{P_t^\varepsilon \},\ \varepsilon \in [0, 1]$ is defined 
as a process such that the transition of $\{P_t^\varepsilon \}$ 
follows $\{P_t^0\}$ with probability $1 - \varepsilon$
and does not follow with probability $\varepsilon$.
Then, we introduce a notion of {\it regular perturbation} as below.

%%%%%%%%%%%%%%%%%%%%%%%%%%%%%%%%%%%%%%%%%%%%%%%%%%%%%%
% Definition 3
%%%%%%%%%%%%%%%%%%%%%%%%%%%%%%%%%%%%%%%%%%%%%%%%%%%%%%
\begin{definition}[Regular Perturbation \cite{Young}]
\label{def:3}
A family of stochastic processes $\{P^\varepsilon_t\}$ is called a regular perturbation of $\{P_t^0\}$ if the
following conditions are satisfied:
\begin{description}
\item[(A1)] For some $\varepsilon^* >0$, the process $\{P^\varepsilon _t\}$ is 
irreducible and aperiodic for all $\varepsilon \in (0,\varepsilon^*]$.
\item[(A2)] Let us denote by $P_{xy}^\varepsilon$
the transition probability from 
$x \in \X$ to $y \in \X$ along with the Markov process $\{P^\varepsilon _t\}$.
Then, 
%
%\begin{equation*}
${\rm lim}_{\varepsilon \rightarrow 0}P_{xy}^\varepsilon = P_{xy}^0$
%\end{equation*}
%
holds for all $x, y\in \X$.
\item[(A3)] 
If $P_{xy}^\varepsilon >0$ for some $\varepsilon$, 
then there exists a real number $\chi (x\rightarrow y)\geq 0$ such that 
\begin{equation}
{\rm lim}_{\varepsilon \rightarrow 0}
\frac{P_{xy}^\varepsilon}{\varepsilon ^{\chi (x\rightarrow y)}}
\in (0,\infty ),
\label{eq:2.3}
\end{equation}
where $\chi (x\rightarrow y)$ is called {\it resistance of transition} from $x$ to $y$.
\end{description}
\end{definition}
%%%%%%%%%%%%%%%%%%%%%%%%%%%%%%%%%%%%%%%%%%%%%%%%%%%%%%
Remark that, from (A1), if $\{P^\varepsilon_t\}$ 
is a regular perturbation of $\{P_t^0\}$, 
then $\{P^\varepsilon_t\}$ has the unique stationary 
distribution $\mu (\varepsilon)$ for each $\varepsilon>0$.

We next introduce the {\it resistance $\lambda(r)$ of a path}
$r$ from $x \in \X$ to $x' \in \X$ along with transitions 
$x^{(0)} = x \rightarrow x^{(2)}\rightarrow\cdots\rightarrow x^{(m)} = 
x'$ as the value satisfying
\begin{equation}
\lim_{\varepsilon\rightarrow 0} \frac{P^\varepsilon(r)}{\varepsilon^{\lambda(r)}} \in (0, \infty),
\label{eq:2.4}
\end{equation}
where $P^\varepsilon(r)$ denotes the probability of the sequence of transitions.
Then, it is easy to confirm that $\lambda(r)$ is simply given by
\begin{equation}
 \lambda(r) = \sum_{i =0}^{m-1}\chi(x^{(i)}\rightarrow x^{(i+1)}).
\label{eq:2.5}
\end{equation}

A state $x \in \X$ 
is said to communicate with state $y \in \X$ if both $x \rightsquigarrow y$ 
and $y \rightsquigarrow x$ hold, where the notation $x \rightsquigarrow y$ 
implies that $y$ is accessible from $x$ i.e.
a process starting at state $x$ has non-zero probability of 
transitioning into $y$ at some point. 
A {\it recurrent communication class} is a class such that 
every pair of states in the class communicates with each other and
no state outside the class is accessible to the class.
Now, let $H_1,\cdots ,H_J$ be recurrent communication classes 
 of Markov process $\{P_t^0\}$.
Then, within each class, there is a path with zero resistance from every state to every other.
In case of a perturbed Markov process $\{P_t^\varepsilon\}$,
there may exist several paths from states in $H_l$ to states in $H_k$
for any two distinct recurrent communication classes $H_l$ and $H_k$.
The minimal resistance among all such paths
is denoted by $\chi_{lk}$.

Let us now define a weighted complete directed graph $G = (\H, \H \times 
\H, \W)$ over the recurrent communication classes $\H = \{H_1,\cdots,H_J\}$,
where the weight $w_{lk}\in \W$ of each edge $(H_l,H_k)$
is equal to the minimal resistance $\chi_{lk}$.
We next define {\it $l$-tree} which is 
a spanning tree over $G$ with a root node $H_l\in \H$.
We also denote by $\G(l)$ the set of all $l$-trees. 
The {\it resistance of an $l$-tree} is the sum of
the weights on all the edges of the tree.
The {\it stochastic potential} of the recurrent communication class $H_l$ is the
minimal resistance among all $l$-trees in $\G(l)$.
We also introduce the 
notion of {\it stochastically stable state} as below.
%%%%%%%%%%%%%%%%%%%%%%%%%%%%%%%%%%%%%%%%%%%%%%%%%%%%%%
% Definition 4
%%%%%%%%%%%%%%%%%%%%%%%%%%%%%%%%%%%%%%%%%%%%%%%%%%%%%%
\begin{definition}[Stochastically Stable State \cite{Young}]
\label{def:4}
A state $x\in \X$ is said to be stochastically stable, if $x$ satisfies 
${\lim}_{\varepsilon\rightarrow 0+}\mu _x(\varepsilon)>0$, where
$\mu_x(\varepsilon)$ is the value of an element of 
stationary distribution $\mu(\varepsilon)$ corresponding to state $x$.
\end{definition}
%%%%%%%%%%%%%%%%%%%%%%%%%%%%%%%%%%%%%%%%%%%%%%%%%%%%%%

Using the above terminologies, we introduce the following well known result which connects
the stochastically stable states and stochastic potential.
%%%%%%%%%%%%%%%%%%%%%%%%%%%%%%%%%%%%%%%%%%%%%%%%%%%%%%
% Proposition 1
%%%%%%%%%%%%%%%%%%%%%%%%%%%%%%%%%%%%%%%%%%%%%%%%%%%%%%
\begin{proposition}\cite{Young}
\label{prop:1}
Let $\{P^\varepsilon_t\}$ be a regular perturbation of $\{P_t^0\}$. 
Then ${\lim}_{\varepsilon\rightarrow 0+}\mu (\varepsilon)$ exists and 
the limiting distribution $\mu(0)$ is a stationary distribution of $\{P_t^0\}$. 
Moreover the stochastically stable states are contained 
in the recurrent communication classes with minimum stochastic potential.
\end{proposition}
%%%%%%%%%%%%%%%%%%%%%%%%%%%%%%%%%%%%%%%%%%%%%%%%%%%%%%

\subsection{Ergodicity}
\label{sec:2.3}

Discrete-time Markov processes can be divided into two types: 
time-homogeneous and time-inhomogeneous, where a Markov process
$\{P_t\}$ is said to be time-homogeneous if
the transition matrix denoted by $P_t$ is independent of
the time and to be a time-inhomogeneous
if it is time dependent.
%If a process $\{P_t\}$ is time homogeneous, we simply write 
%the process as $\{P\}$ omitting the subscript $t$ in the following.
We also denote the probability of the state transition from
time $k_0$ to time $k$ by $P(k_0,k) = \prod^{k-1}_{t=k_0}P_t,\  0 \leq k_0 
< k$.

For a Markov process $\{P_t\}$, we introduce the notion of ergodicity.
%%%%%%%%%%%%%%%%%%%%%%%%%%%%%%%%%%%%%%%%%%%%%%%%%%%%%%
% Definition 5
%%%%%%%%%%%%%%%%%%%%%%%%%%%%%%%%%%%%%%%%%%%%%%%%%%%%%%
\begin{definition}[Strong Ergodicity \cite{er}]
\label{def:5}
A Markov process $\{P_t\}$ is said to be strongly ergodic 
if there exists a stochastic vector $\mu ^*$ such 
that for any distribution $\mu$ on $\X$ and time $k_0$, 
we have ${\rm lim}_{k\rightarrow \infty}\mu P(k_0,k) = \mu^*$.
\end{definition}
%%%%%%%%%%%%%%%%%%%%%%%%%%%%%%%%%%%%%%%%%%%%%%%%%%%%%%
%%%%%%%%%%%%%%%%%%%%%%%%%%%%%%%%%%%%%%%%%%%%%%%%%%%%%%
% Definition 6
%%%%%%%%%%%%%%%%%%%%%%%%%%%%%%%%%%%%%%%%%%%%%%%%%%%%%%
\begin{definition}[Weak Ergodicity \cite{er}]
A Markov process $\{P_t\}$ is said to be weakly ergodic if  
the following equation holds.
\begin{equation*}
\lim_{k\rightarrow \infty}(P_{xz}(k_0,k)-P_{yz}(k_0,k)) = 0\ \ 
{\forall x},y,z \in \X,\ {\forall k_0}\in {\mathbb Z}_+
\end{equation*}
\end{definition}
%%%%%%%%%%%%%%%%%%%%%%%%%%%%%%%%%%%%%%%%%%%%%%%%%%%%%%
If $\{P_t\}$ is strongly ergodic, the distribution $\mu $ converges 
to the unique distribution $\mu^*$ from any initial state.
Weak ergodicity implies that the information on
the initial state vanishes as time increases though
convergence of $\mu$ may not be guaranteed. 
Note that the notions of weak and strong ergodicity
are equivalent in case of time-homogeneous Markov processes.

We finally introduce the following well-known results on ergodicity.
%%%%%%%%%%%%%%%%%%%%%%%%%%%%%%%%%%%%%%%%%%%%%%%%%%%%%%
% Proposition 2
%%%%%%%%%%%%%%%%%%%%%%%%%%%%%%%%%%%%%%%%%%%%%%%%%%%%%%
\begin{proposition}\cite{er}
\label{prop:2}
A Markov process $\{P_t\}$ is strongly ergodic if the following conditions hold:
\begin{description}
\item[(B1)] The Markov process $\{P_t\}$ is weakly ergodic.
\item[(B2)] For each $t$, there exists a stochastic vector 
$\mu ^{t}$ on $\X$ such that $\mu ^{t}$ is the left
eigenvector of the transition matrix $P(t)$ with eigenvalue 1.
\item[(B3)] The eigenvector $\mu^{t}$ in (B2) satisfies 
$\sum^{\infty}_{t=0}\sum_{x\in \X}| \mu^{t}_{x} - \mu^{t+1}_{x} | < \infty $.
Moreover, if $\mu ^{*}=\lim_{t\rightarrow \infty }\mu^{t}$, then $\mu ^{*}$ is 
the vector in Definition \ref{def:5}.
\end{description}
\end{proposition}
%%%%%%%%%%%%%%%%%%%%%%%%%%%%%%%%%%%%%%%%%%%%%%%%%%%%%%

\section{Learning Algorithm and Main Result}
\label{sec:3}

In this section, we present a learning algorithm
called {\it Payoff-based Inhomogeneous Partially Irrational Play (PIPIP)}
and state the main result of this paper.
At each iteration $t\in {\mathbb Z}_+$, the learning algorithm chooses
an action according to the following procedure assuming that
each agent $i\in \V$ stores previous two chosen actions 
$a_i(t-2), a_i(t-1)$ and the outcomes $U_i(a(t-2)), U_i(a(t-1))$.
Each agent first updates a parameter
$\varepsilon$ called {\it exploration rate} by
\begin{equation}
\varepsilon(t)=t^{-\frac{1}{n(D+1)}},
\label{eq:3.1}
\end{equation}
where $D$ is defined as $D:={\max}_{i\in {\mathcal V}} D_i$ and
$D_i$ is the minimal number of steps required for
transitioning between any two actions of agent $i$.

Then, each agent compares the values of $U_i(a(t-1))$ and $U_i(a(t-2))$.
If $U_i(a(t-1))\geq U_i(a(t-2))$ holds, then
he chooses action $a_i(t)$ according to the rule:
\begin{itemize}
\item $a_i(t)$ is randomly chosen from $\R_i(a_i(t-1))\setminus \{a_i(t-1)\}$
with probability $\varepsilon(t)$,
(it is called an {\it exploratory} decision).
\item $a_i(t)=a_i(t-1)$ with probability $1-\varepsilon(t)$.
\end{itemize}
Otherwise ($U_i(a(t-1))< U_i(a(t-2))$), action $a_i(t)$ is chosen
according to the rule:
\begin{itemize}
\item $a_i(t)$ is randomly chosen from $\R_i(a_i(t-1))
\setminus \{a_i(t-1), a_i(t-2)\}$  with probability $\varepsilon(t)$
(it is called an {\it exploratory} decision).
\item  $a_i(t) = a_i(t-1)$
with probability 
\begin{equation}
 (1-\varepsilon(t))(\kappa\cdot \varepsilon
      (t)^{\Delta_i}),\ \Delta_i := U_i(a(t-2))-U_i(a(t-1))
\label{eq:3.2}
\end{equation}
(it is called an {\it irrational} decision).
\item  $a_i(t)=a_i(t-2)$ with probability 
\begin{equation}
(1-\varepsilon(t))(1-\kappa\cdot \varepsilon(t)^{\Delta_i}).
\label{eq:3.3}
\end{equation}
\end{itemize}
The parameter $\kappa$ should be chosen so as to satisfy
\begin{equation}
\kappa\in \Big(\frac{1}{C -1}, \frac{1}{2}\Big],\ 
C := \max_{i\in {\mathcal V}}\max_{a_i\in {\mathcal A}_i}|\R_i(a_i)|,
\label{eq:3.4}
\end{equation}
where $|\R_i(a_i)|$ is the number of elements of the set $\R_i(a_i)$.
It is clear under the third item of Assumption \ref{ass:1} that 
the action $a_i(t)$ is well-defined.

Finally, each agent $i$ executes the selected action $a_i(t)$ and computes
the resulting utility $U_i(a(t))$ via feedbacks from environment and neighboring agents.
At the next iteration, agents repeat the same procedure.

The algorithm PIPIP is compactly described in Algorithm \ref{alg:1},
where the function ${\rm rnd}({\mathcal A}')$ outputs an action
chosen randomly from the set ${\mathcal A}'$.
Note that the algorithm with a constant $\varepsilon(t)=\varepsilon\in (0,1/2]$
is called {\it Payoff-based Homogeneous Partially Irrational Play (PHPIP)},
which will be used for the proof of the main result of this paper.

%%%%%%%%%%%%%%%%%%%%%%%%%%%%%%%%%%%%%%%%%%%%%%%%%%%%%%
% Algorithm 1
%%%%%%%%%%%%%%%%%%%%%%%%%%%%%%%%%%%%%%%%%%%%%%%%%%%%%%
\begin{algorithm}[t]
\caption{Payoff-based Inhomogeneous Partially Irrational Play (PIPIP)}
\label{alg:1}
\begin{algorithmic}
\STATE {\bf Initialization:}
Action $a$ is chosen randomly from $\A$.
Set $a_i^1 \leftarrow a_i,\  a_i^2\leftarrow a_i,\
U_i^1\leftarrow U_i(a),\ U_i^2 \leftarrow U_i(a)$,
$\Delta_i \leftarrow 0$ for all $i\in \V$ and $t \leftarrow 2$.
\STATE {\bf Step 1:} $\varepsilon\leftarrow t^{(-1/(n(D+1)))}$. 
\STATE {\bf Step 2:} 
If $U_i^1\geq U_i^2$, then set
\begin{equation*}
 a_i^{tmp} \leftarrow \left\{
\begin{array}{ll}
{\rm rnd}(\R_i(a_i^1)\setminus \{a_i^1\}), & \mbox{w.p. } 
\varepsilon\\
a_i^1,& \mbox{w.p. }1-\varepsilon
\end{array}
\right..
\end{equation*}
Otherwise, set
\begin{equation*}
 a_i^{tmp} \leftarrow \left\{
\begin{array}{ll}
{\rm rnd}(\R_i(a_i^1)\setminus \{a_i^1,a_i^2\}), & \mbox{w.p. } 
\varepsilon(t)\\
a_i^1,& \mbox{w.p. }(1-\varepsilon)(\kappa\cdot \varepsilon^{\Delta_i})\\
a_i^2,& \mbox{w.p. }
(1-\varepsilon)(1-\kappa\cdot \varepsilon^{\Delta_i})
\end{array}
\right..
\end{equation*}
\STATE {\bf Step 3:} Execute the selected action $a_i^{tmp}$
and receive $U_i^{tmp} \leftarrow U_i(a^{tmp})$. 
\STATE {\bf Step 4:} Set
$a_i^2 \leftarrow a^1_i,\ a_i^1 \leftarrow a_i^{tmp},\ U_i^2 \leftarrow U_i^1,\
U_i^1 \leftarrow U_i^{tmp}$ and 
$\Delta_i \leftarrow U_i^2 - U_i^1$.
\STATE {\bf Step 5:} 
$t \leftarrow t+1$ and go to {\bf Step 1}. 
\end{algorithmic}
\end{algorithm}
%%%%%%%%%%%%%%%%%%%%%%%%%%%%%%%%%%%%%%%%%%%%%%%%%%%%%%

PIPIP is developed based on the learning algorithm
DISL presented in \cite{ZM09a}.
The main difference of PIPIP from DISL 
is that agents may choose the 
action with the lower utility in {\bf Step 2} with probability 
$(1-\varepsilon)(\kappa\cdot \varepsilon^{\Delta_i})$
which depends on the difference of the last two steps' utilities $\Delta_i$
and the parameters $\kappa$ and $\varepsilon$. 
Thanks to the irrational decisions, agents can escape from
undesirable Nash equilibria as will be proved in the next section.

We are now ready to state the main result of this paper.
Before mentioning it, we define 
\begin{equation}
{\mathcal B}:=\{(a,a')\in {\mathcal A}\times {\mathcal A}|\ a'_i\in 
\R_i(a_i)\ \forall i \in {\mathcal V}\}.
\label{eq:3.5}
\end{equation}
and $\zeta (\Gamma)$ as
the set of the optimal Nash equilibria, i.e. potential function maximizers, 
of a constrained potential game $\Gamma$.

%%%%%%%%%%%%%%%%%%%%%%%%%%%%%%%%%%%%%%%%%%%%%%%%%%%%%%
% Theorem 1
%%%%%%%%%%%%%%%%%%%%%%%%%%%%%%%%%%%%%%%%%%%%%%%%%%%%%%
\begin{theorem}\label{thm:1}
Consider a constrained potential game 
$\Gamma$ satisfying Assumptions \ref{ass:1} and \ref{ass:2}.
Suppose that each agent behaves according to Algorithm \ref{alg:1}.
Then, a Markov process $\{P_t\}$ is defined over the space ${\mathcal B}$
and the following equation is satisfied.
\begin{equation}
\lim_{t\rightarrow \infty}{\rm Prob}\left[z(t)\in 
{\rm diag}\left(\zeta(\Gamma)\right)\right]=1,
\label{eq:3.6}
\end{equation}
where $z(t) := (a(t-1),a(t))$ and
$\diag(\A') = \{(a,a)\in \A \times \A|\ a\in {\mathcal A}'\},\ \A' \subseteq \A$.
\end{theorem}
%%%%%%%%%%%%%%%%%%%%%%%%%%%%%%%%%%%%%%%%%%%%%%%%%%%%%%
Equation (\ref{eq:3.6}) means that the probability that 
agents executing PIPIP take 
one of the potential function maximizers converge to $1$.
The proof of this theorem will be shown in the next section.

In PIPIP, the parameter $\varepsilon(t)$ is updated by (\ref{eq:3.1})
 to prove the above theorem, which is the same as DISL.
However, this update rule takes long time to
reach a sufficiently small $\varepsilon(t)$ when
the size of the game, i.e. $n(D+1)$ is large.
Thus, from the practical point of view, we might 
have to decrease $\varepsilon(t)$ based on heuristics
 or use PHPIP with a sufficiently small $\varepsilon$.
Even in such cases, the following theorem at least holds 
similarly to the paper \cite{JYGS09}.
%%%%%%%%%%%%%%%%%%%%%%%%%%%%%%%%%%%%%%%%%%%%%%%%%%%%%%
% Theorem 2
%%%%%%%%%%%%%%%%%%%%%%%%%%%%%%%%%%%%%%%%%%%%%%%%%%%%%%
\begin{theorem}
\label{thm:2}
Consider a constrained potential game 
$\Gamma$ satisfying Assumptions \ref{ass:1} and \ref{ass:2}.
Suppose that each agent behaves according to PHPIP.
Then, given any probability $p<1$, %(this value becomes large as $\varepsilon$ decreases),
if the exploration rate $\varepsilon$ is sufficiently small, 
for all sufficiently large time $t\in {\mathbb Z}_+$, the following equation holds.
\begin{equation}
{\rm Prob} \left[z(t)\in {\rm diag}\left(\zeta(\Gamma)\right)\right] > 
 p.
\label{eq:3.7}
\end{equation}
\end{theorem}
%%%%%%%%%%%%%%%%%%%%%%%%%%%%%%%%%%%%%%%%%%%%%%%%%%%%%%
Theorem \ref{thm:2} assures that 
the optimal actions are eventually selected with high probability
as long as the final value of $\varepsilon(t)$ is sufficiently small
irrespective of the decay rate of $\varepsilon(t)$.

\section{Proof of Main Result}
\label{sec:4}

In this section, we prove the main result of this paper (Theorem \ref{thm:1}).
We first consider PHPIP with a constant exploration rate $\varepsilon$.
The state $z(t) = (a(t-1),a(t))$ for PHPIP with $\varepsilon$ constitutes 
a perturbed Markov process $\{P_t^\varepsilon\}$ on the state space 
${\mathcal B}=\{(a,a')\in {\mathcal A}\times {\mathcal A}|\ a'_i\in 
\R_i(a_i)\ \forall i \in {\mathcal V}\}$.

In terms of the Markov process $\{P_t^\varepsilon\}$ induced by PHPIP,
the following lemma holds.
%%%%%%%%%%%%%%%%%%%%%%%%%%%%%%%%%%%%%%%%%%%%%%%%%%%%%%
% Lemma 1
%%%%%%%%%%%%%%%%%%%%%%%%%%%%%%%%%%%%%%%%%%%%%%%%%%%%%%
\begin{lemma}\label{lem:1}
The Markov process $\{P_t^\varepsilon\}$ induced by PHPIP applied to
a constrained potential game $\Gamma$ is 
a regular perturbation of $\{P_t^0 \}$ under Assumption \ref{ass:1}.
\end{lemma}
%%%%%%%%%%%%%%%%%%%%%%%%%%%%%%%%%%%%%%%%%%%%%%%%%%%%%%
%
\begin{proof}
Consider a feasible transition $z^1\rightarrow z^2$ with 
$z^1=(a^0,a^1)\in \BB$ and $z^2=(a^1,a^2)\in \BB$ and
partition the set of agents ${\mathcal V}$ 
according to their behaviors along with the transition as 
\begin{eqnarray}
&&\Lambda _1=\{i\in {\mathcal V}|\ U_i(a^1)\geq U_i(a^0),\ a^2_i\in \R_i(a_i^1)\setminus \{a_i^1\}\},
\nonumber\\
&&
\Lambda _2=\{i\in {\mathcal V}|\ U_i(a^1)\geq U_i(a^0),\ a^2_i=a_i^1\},
\nonumber\\
&& \Lambda _3=\{i\in {\mathcal V}|\ U_i(a^1)< U_i(a^0),\ a^2_i\in \R_i(a_i^1)\setminus \{a_i^0, a_i^1\}\},
\nonumber\\
&&\Lambda _4=\{i\in {\mathcal V}|\ U_i(a^1)< U_i(a^0),\ a^2_i= a_i^1\},
\nonumber\\
&&
\Lambda _5=\{i\in {\mathcal V}|\ U_i(a^1)< U_i(a^0),\ a^2_i= a_i^0\}.
\nonumber
\end{eqnarray}
Then, the probability of the transition $z^1\rightarrow z^2$
is represented as
\begin{eqnarray}
P_{z^1z^2}^\varepsilon&=& \prod_{i\in \Lambda _1} 
\frac{\varepsilon}{| \R_i(a_i^1)| -1}\times \prod_{i\in \Lambda _2} 
(1-\varepsilon)\times \prod_{i\in \Lambda _3} \frac {\varepsilon}
{| \R_i(a_i^1)| - {h_i}}\times  
\nonumber\\
&&\times \prod_{i\in \Lambda _4}(1-\varepsilon)\kappa\varepsilon^{\Delta _i}
\times \prod_{i\in \Lambda_5}(1-\varepsilon)(1-\kappa\varepsilon^{\Delta _i}),
\label{eq:4.1}
\end{eqnarray}
where $h_i=1$ if $a_i^0=a_i^1$ and $h_i=2$ otherwise.
We see from (\ref{eq:4.1}) that
the resistance of transition $z^1\rightarrow z^2$ defined in (\ref{eq:2.3}) is given by 
$ |\Lambda _1| +|\Lambda _3| +\sum_{i\in \Lambda _4}\Delta _i$ since   
\begin{equation}
 0<\lim_{\varepsilon\rightarrow 0} 
\frac {P_{z_1z_2}^\varepsilon}
{\varepsilon^{|\Lambda_1| +|\Lambda _3| +\sum_{i\in \Lambda_4}\Delta _i}}
=\prod _{i\in \Lambda_1}\frac {1}{|\R_i(a_i^1)|-1}\prod _{i\in \Lambda_3}
\frac {1}{|\R_i(a_i^1)| - h_i}\times \kappa^{|\Lambda_4|}<+\infty
\label{eq:4.2}
\end{equation}
holds.
Thus, {\bf (A3)} in Definition \ref{def:3} is satisfied. 
In addition, it is straightforward from the procedure of PHPIP to confirm
the condition {\bf (A2)}.

It is thus sufficient to check {\bf (A1)} in Definition \ref{def:3}. 
From the rule of taking exploratory actions in 
Algorithm 1 and the second item of Assumption \ref{ass:1},
we immediately see that the set of the states accessible 
from any $z\in {\mathcal B}$ is equal to ${\mathcal B}$. 
This implies that the perturbed Markov process 
$\{P_t^\varepsilon\}$ is irreducible.
We next check aperiodicity of $\{P_t^\varepsilon\}$. 
It is clear that 
any state in ${\rm diag}({\mathcal A})=\{(a,a)\in \A \times \A|\
 a\in {\mathcal A}\}$ has period 1. 
Let us next pick any 
$(a^0,a^1)$ from the set ${\mathcal B}\setminus {\rm diag}({\mathcal A})$.
Since $a_i^0\in \R_i(a_i^1)$ holds iff $a_i^1\in \R_i(a_i^0)$
from Assumption \ref{ass:1},
the following two paths are both feasible:
$(a^0,a^1)\rightarrow (a^1,a^0)\rightarrow (a^0,a^1)$, 
$(a^0,a^1)\rightarrow (a^1,a^1)\rightarrow (a^1,a^0)\rightarrow (a^0,a^1)$.
This implies that the period of state $(a^0,a^1)$ is $1$ and
the process $\{P_t^\varepsilon\}$ is proved to be aperiodic.
Hence the process $\{P_t^\varepsilon\}$ is both 
irreducible and aperiodic, which means {\bf (A1)} in Definition \ref{def:3}.

In summary, conditions {\bf (A1)}--{\bf (A3)} 
in Definition \ref{def:3} are satisfied
and the proof is completed.
\end{proof}
From Lemma \ref{lem:1}, the perturbed Markov process
$\{P_t^\varepsilon\}$ is irreducible and hence there exists 
a unique stationary distribution $\mu (\varepsilon)$ for every 
$\varepsilon$. 
Moreover, because $\{P_t^\varepsilon\}$ is a regular perturbation of 
$\{P_t^0 \}$, we see from the former half of Proposition \ref{prop:1} that 
$\lim_{\varepsilon\rightarrow 0+}\mu (\varepsilon)$ exists and 
the limiting distribution $\mu (0)$ is the stationary distribution of $\{P_t^0\}$.

We also have the following lemma on the Markov process $\{P_t^\varepsilon\}$ induced by PHPIP.
%%%%%%%%%%%%%%%%%%%%%%%%%%%%%%%%%%%%%%%%%%%%%%%%%%%%%%
% Lemma 2
%%%%%%%%%%%%%%%%%%%%%%%%%%%%%%%%%%%%%%%%%%%%%%%%%%%%%%
\begin{lemma}
\label{lem:2}
Consider the Markov process $\{P_t^\varepsilon\}$ induced by PHPIP
applied to a constrained potential game $\Gamma$.
Then, the recurrent communication classes $\{\H_i\}$
of the unperturbed Markov process $\{P_t^0 \}$ are given by elements of 
$\diag(\A)=\{(a,a)\in \A \times \A|\ a\in {\mathcal A}\}$, namely 
\begin{equation}
 \H_i = \{(a^i,a^i)\},\ i\in 1,\cdots,|\A|.
\label{eq:4.3}
\end{equation}
\end{lemma}
%%%%%%%%%%%%%%%%%%%%%%%%%%%%%%%%%%%%%%%%%%%%%%%%%%%%%%
%
\begin{proof}
Because of the rule at Step 2 of PHPIP, it is clear that 
any state belonging to $\diag( {\mathcal A})$ cannot move to another
state without explorations, which implies that 
all the states in $\diag(\A)$ itself form recurrent communication 
classes of the unperturbed Markov process $\{P_t^0 \}$.

We next consider the states in ${\mathcal B}\setminus \diag({\mathcal A})$
and prove that these states are never included in recurrent communication 
classes of the unperturbed Markov process $\{P_t^0\}$.
Here, we use induction.
We first consider the case of $n=1$.
If $U_1(a_1^1)\geq U_1(a_1^0)$, then 
the transition $(a_1^0,a_1^1)\rightarrow (a_1^1,a_1^1)$ is taken.
Otherwise, a sequence of transitions 
$(a_1^0,a_1^1)\rightarrow (a_1^1,a_1^0) \rightarrow (a_1^0,a_1^0)$ occurs.
Thus, in case of $n = 1$, the state 
$(a_1^0,a_1^1)\in {\mathcal  B}\setminus\diag({\mathcal A})$ is never included 
in recurrent communication classes of $\{P_t^0 \}$.

We next make a hypothesis that there exists a $k \in {\mathbb Z}_+$ such 
that all the states in ${\mathcal B}\setminus \diag({\mathcal A})$ are not 
included in recurrent communication classes of 
the unperturbed Markov process $\{P_t^0\}$ for all $n\leq k$.
Then, we consider the case $n = k+1$, where
there are three possible cases:
\begin{description}
\item[(i)] $U_i(a^1)\geq U_i(a^0)\ {\forall i}\in \V = \{1,\cdots, k+1\}$,
\item[(ii)] $U_i(a^1)< U_i(a^0)\ {\forall i}\in \V = \{1,\cdots, k+1\}$,
\item[(iii)] $U_i(a^1)\geq U_i(a^0)$ for $l$ agents where
$l \in \{2,\cdots, k\}$.
\end{description}
In case (i), the transition 
$(a^0,a^1)\rightarrow (a^1,a^1)$ must occur for $\varepsilon= 0$ and,
%Namely, such states are not included in the recurrent communication classes.
in case (ii), the transition
$(a^0,a^1)\rightarrow (a^1,a^0)\rightarrow (a^0,a^0)$ should be selected.
Thus, all the states in ${\mathcal B}\setminus \diag({\mathcal A})$
satisfying (i) or (ii) are 
never included in recurrent communication classes.
In case (iii), at the next iteration, 
all the agents $i$ satisfying $U_i(a^1)\geq U_i(a^0)$
choose the current action.
Then, such agents possess a single action
in the memory and, 
in case of $\varepsilon= 0$, each agent has to choose either of
the actions in the memory.
Namely, these agents never change their actions in all subsequent iterations.
The resulting situation is thus the same as the case of $n = k + 1 - l$.
From the above hypothesis, we can conclude that
the states in case (iii) are also not included in recurrent communication classes. 
In summary, the states in
${\mathcal B}\setminus \diag({\mathcal A})$ are never included 
in the recurrent communication classes of $\{P_t^0 \}$.
The proof is thus completed.
\end{proof}

A feasible path over the process
$\{P_t^\varepsilon\}$ from $z\in {\mathcal B}$ to $z'\in {\mathcal B}$ is especially 
said to be a {\it route}
if both of the two nodes $z$ and $z'$ 
are elements of $\diag(\A) \subset {\mathcal B}$.
Note that a route is a path and hence the resistance of the route
is also given by (\ref{eq:2.4}).
Especially, we define a {\it straight route} as follows, where we use 
the notation
\begin{eqnarray}
&&\E_{single} := \{(z=(a,a),z'=(a',a'))\in \diag(\A)\times\diag(\A)|\ 
\nonumber\\
&&\hspace{5cm}\exists i \in \V \mbox{ s.t. }a_i\in \R_i(a_i'), a_i \neq a_i' \mbox{ and } a_{-i} 
= a_{-i}'\}.
\label{eq:4.4}
\end{eqnarray}
%
%%%%%%%%%%%%%%%%%%%%%%%%%%%%%%%%%%%%%%%%%%%%%%%%%%%%%%
% Definition 7
%%%%%%%%%%%%%%%%%%%%%%%%%%%%%%%%%%%%%%%%%%%%%%%%%%%%%%
\begin{definition}[Straight Route]
\label{def:7}
A route between any two states $z^0=(a^0,a^0)$ and 
 $z^1=(a^1,a^1)$ in ${\diag}(\A)$ such that
$(z^0,z^1)\in \E_{single}$
is said to be a straight route if the path is given by the transitions
on the Markov process $\{P_t^\varepsilon\}$
such that only one agent $i$ changes his action from $a_i^0$ to $a_i^1$
at first iteration and the explored agent $i$ selects the same action 
$a_i^1$ at the next iteration while the other agents choose the same action 
$a_{-i}^0=a_{-i}^1$ during the two steps.
\end{definition}
%%%%%%%%%%%%%%%%%%%%%%%%%%%%%%%%%%%%%%%%%%%%%%%%%%%%%%
In terms of the straight route, we have the following lemma.
%%%%%%%%%%%%%%%%%%%%%%%%%%%%%%%%%%%%%%%%%%%%%%%%%%%%%%
% Lemma 3
%%%%%%%%%%%%%%%%%%%%%%%%%%%%%%%%%%%%%%%%%%%%%%%%%%%%%%
\begin{lemma}
\label{lem:3}
Consider paths from any state $z^0=(a^0,a^0)\in \diag(\A)$ to any state 
$z^1=(a^1,a^1)\in \diag(\A)$ such that $(z^0,z^1) \in \E_{single}$
over the Markov process $\{P_t^\varepsilon\}$ induced by PHPIP applied to
a constrained potential game $\Gamma$.
Then, under Assumption \ref{ass:2}, the resistance $\lambda(r)$ of 
the straight route $r$ from $z^0$ to $z^1$ is strictly 
smaller than $2$ and the resistance $\lambda(r)$ is minimal
among all paths from $z^0$ to $z^1$.
\end{lemma}
%%%%%%%%%%%%%%%%%%%%%%%%%%%%%%%%%%%%%%%%%%%%%%%%%%%%%%
%
\begin{proof}
Along with the straight route, only one agent $i$ first 
changes action from $a_i^0$ to 
$a_i^1$, whose probability is given by
\begin {equation}
(1-\varepsilon)^{n-1}\times \frac{\varepsilon}{|\R_i(a_i^0)| - 1}.
\label{eq:4.5}
\end {equation}
It is easy to confirm from (\ref{eq:4.5})
that the resistance of the transition 
$(a^0,a^0)\rightarrow (a^0,a^1)$ is equal to $1$.
We next consider the transition from $(a^0,a^1)$ to $(a^1,a^1)$.
If $U_i(a^1)\geq U_i(a^0)$ is true, the probability of this transition 
is given by $(1-\varepsilon)^{n}$,
whose resistance is equal to $0$.
Otherwise, $U_i(a^1)<U_i(a^0)$ holds and the 
probability of this transition is given by  
$(1-\varepsilon)^{n}\times \kappa\varepsilon^{\Delta_i}$,
whose resistance is $\Delta_i$.
Let us now notice that the resistance $\lambda(r)$ 
of the straight route $r$ is equal to the sum of the resistances 
of transitions $(a^0,a^0)\rightarrow (a^0,a^1)$ and $(a^0,a^1)\rightarrow 
 (a^1,a^1)$ from (\ref{eq:2.5})
and that $\Delta_i < 1$ from Assumption \ref{ass:2}.
We can thus conclude that $\lambda(r)$ is smaller than $2$.
It is also easy to confirm that the resistance of paths
such that more than $1$ agents take
exploratory action should be greater than $2$.
Namely, the straight route gives the smallest
resistance among all paths from $z^0 = (a^0,a^0)$ to $z^1 = (a^1,a^1)$
and hence the proof is completed.
\end{proof}

We also introduce the following notion.
%%%%%%%%%%%%%%%%%%%%%%%%%%%%%%%%%%%%%%%%%%%%%%%%%%%%%%
% Definition 8
%%%%%%%%%%%%%%%%%%%%%%%%%%%%%%%%%%%%%%%%%%%%%%%%%%%%%%
\begin{definition}[$m$-Straight-Route]
\label{def:8}
An $m$-straight-route is a route which passes through 
$m$ vertices in $\diag({\mathcal A})$ and
all the routes between any two of these vertices are straight. 
\end{definition}
%%%%%%%%%%%%%%%%%%%%%%%%%%%%%%%%%%%%%%%%%%%%%%%%%%%%%%
In terms of the route,
we can prove the following lemma, which
clarifies a connection
between the potential function and the resistance of the route.
%%%%%%%%%%%%%%%%%%%%%%%%%%%%%%%%%%%%%%%%%%%%%%%%%%%%%%
% Lemma 4
%%%%%%%%%%%%%%%%%%%%%%%%%%%%%%%%%%%%%%%%%%%%%%%%%%%%%%
\begin{lemma}\label{lem:4}
Consider the Markov process $\{P_t^\varepsilon\}$ induced by PHPIP
applied to a constrained potential game $\Gamma$.
Let us denote an $m$-straight-route $r$ over $\{P_t^\varepsilon\}$ 
from state $z^0 = (a^0,a^0)\in \diag({\mathcal A})$ to state 
$z^1=(a^1,a^1)\in {\diag}({\mathcal A})$ by
\begin {equation}
z^{(0)} = z^0{\Rightarrow}z^{(1)}\Rightarrow z^{(2)}{\Rightarrow}z^{(3)}
{\Rightarrow}\cdots z^{(m-3)}\Rightarrow z^{(m-2)}{\Rightarrow}
z^{(m-1)} = z^1, 
\label{eq:4.6}
\end {equation}
where $z^{(i)} = (a^{(i)},a^{(i)})\in \diag(\A), i\in \{0,\cdots, m-1\}$ and all the arrows
between them are straight routes.
In addition, we denote its reverse route $r'$ by
\begin {equation}
z^{(0)} = z^0{\Leftarrow}z^{(1)}{\Leftarrow}z^{(2)}{\Leftarrow}z^{(3)}{\Leftarrow}\cdots
{\Leftarrow} z^{(m-3)}{\Leftarrow}z^{(m-2)}{\Leftarrow}z^{(m-1)} = z^1, 
\label{eq:4.7}
\end {equation}
which is also an $m$-straight route from $z^0$ to $z^1$.
Then, under Assumption \ref{ass:2}, if $\phi(a^0) > \phi(a^1)$, 
we have $\lambda(r) > \lambda(r')$.
\end{lemma}
%%%%%%%%%%%%%%%%%%%%%%%%%%%%%%%%%%%%%%%%%%%%%%%%%%%%%%
%
\begin{proof}
We suppose that the route $r$ contains $p$ straight routes with
resistance greater than $1$ and $r'$ contains $q$ straight routes 
with resistance greater than $1$. 
Let us now denote the explored agent along with the route
$z^{(i)}\Rightarrow z^{(i+1)}$
by $j_i$ and that with $z^{(i)}\Leftarrow z^{(i+1)}$ by $j'_{i}$. 
%It is also clear that either of $U_{j_i}(a^{(i+1)})\geq U_{j_i}(a^{(i)})$ and 
%$U_{j_i}(a^{(i+1)})<U_{j_i}(a^{(i)})$ holds for all $i\in \{0,\cdots, m-1\}$.
From the proof of Lemma \ref{lem:3},
the resistance of the route $z^{(i)}\Rightarrow z^{(i+1)}$ 
should be exactly equal to $1$
(in case of $U_{j_i}(a^{(i+1)})\geq U_{j_i}(a^{(i)})$)
or equal to $1 + \Delta_{j_i} \in (1, 2)$ (in case of 
$U_{j_i}(a^{(i+1)})<U_{j_i}(a^{(i)})$).
From (\ref{eq:2.1}), the following equation holds.
\begin {equation}
\Delta_{j_i} = U_{j_i}(a^{(i)}) - U_{j_i}(a^{(i+1)})
= \phi(a^{(i)}) - \phi(a^{(i+1)}) = 
U_{j_i'}(a^{(i)}) - U_{j_i'}(a^{(i+1)}) = - \Delta_{j_i'}.
\label{eq:4.8}
\end {equation}
Namely, one of the resistances of the straight routes 
$z^{(i)} \Rightarrow z^{(i+1)}$ and $z^{(i+1)} \Leftarrow z^{(i)}$
is exactly $1$ and the other is greater than $1$
except for the case that $U_i(a^{(i+1)}) = U_i(a^{(i)})$
in which the resistances are both equal to $1$.
An illustrative example of the relation is given as follows, where
the numbers put on arrows are 
the resistances of the routes.
\begin{equation}
z^{(0)} = z^0 \stackrel{1+\Delta_{j_0}}{\Rightarrow} z^{(1)} 
\stackrel{1}{\Rightarrow} z^{(2)} 
\stackrel{1 + \Delta_{j_1}}{\Rightarrow} z^{(3)}
\stackrel{1}{\Rightarrow} \cdots
\stackrel{1}{\Rightarrow}z^{(m-3)}
\stackrel{1}{\Rightarrow}z^{(m-2)}
\stackrel{1 + \Delta_{j_{m-2}}}{\Rightarrow}z^{(m-1)} = z^1 
\nonumber
\end{equation}
\begin {equation}
z^{(0)} = z^0\stackrel{1}{\Leftarrow}z^{(1)}
\stackrel{1 + \Delta_{j'_{1}}}{\Leftarrow}z^{(2)}
\stackrel{1}{\Leftarrow }z^{(3)}
\stackrel{1 + \Delta_{j'_{3}}}{\Leftarrow} \cdots
\stackrel{1 + \Delta_{j_{m-4}'}}{\Leftarrow}z^{(m-3)}
\stackrel{1 + \Delta_{j_{m-3}'}}{\Leftarrow}z^{(m-2)}
\stackrel{1}{\Leftarrow }z^{(m-1)} = z^1 
\nonumber
\end {equation}
Namely, the inequality $p + q \leq m-1$ holds true.
Let us now collect all the $\Delta_{j_i}$ such that
the resistance
of $z^{(i)}\Rightarrow z^{(i+1)}$ is greater than $1$
and number them as
$\Delta_1, \Delta_2, \cdots, \Delta_{p}$.
Similarly, we define $\Delta'_1, \Delta'_2, \cdots, \Delta'_{q}$
for the reverse route $r'$.
Then, from equations in (\ref{eq:4.8}), we obtain
\begin{equation}
\Delta_1 + \Delta_2 + \cdots + \Delta_{p} - 
(\Delta'_1 + \Delta'_2 + \cdots + \Delta'_{q}) 
= \phi(a^0)-\phi(a^1).
\label{eq:4.9}
\end{equation}
Note that (\ref{eq:4.9}) holds
even in the presence of pairs $(a^{(i)},a^{(i+1)})$ such that
$U_{j_i}(a^{(i+1)}) = U_{j_i}(a^{(i)})$.
Since $\Delta_1+\cdots+\Delta_p=\lambda(r)-(m-1)$ 
and $\Delta'_1+\cdots+\Delta'_q=\lambda(r')-(m-1)$ 
from (\ref{eq:2.5}), we obtain
\begin {equation}
\lambda(r) = \lambda(r') + \phi(a^0) - \phi(a^1).
\label{eq:4.10}
\end {equation}
It is straightforward from (\ref{eq:4.10}) to prove
the statement in the lemma. 
\end{proof}

Let us form the weighted digraph $G$ 
over the recurrent communication classes
for the Markov process $\{P_t^\varepsilon\}$
induced by PHPIP as in Subsection \ref{sec:2.2},
where the weight $w_{lk}$ of each edge $(H_l,H_k)$
is equal to the minimal resistance $\chi_{lk}$ among all the paths connecting 
two recurrent communication classes $H_l$ and $H_k$.
From Lemma \ref{lem:2}, the nodes of the graph $G$ are
given by each element of the set $\diag(\A)$
and hence $G = (\diag(\A),\E,\W), \E \subseteq \diag(\A) \times \diag(\A)$.
Since all the recurrent communication classes have only one element as in (\ref{eq:4.3}),
the weight $w_{lk}$ for any two states $z^l, z^k\in \diag(\A)$
is simply given by the path with minimal resistance 
among all paths from $z^l$ to $z^k$.
In addition, Lemma \ref{lem:3} proves that
if $(z^l,z^k)\in \E_{single}$, the weight $w_{lk} = \chi_{lk}$ 
is given by the resistance of the straight route
from $z^l$ to $z^k$.

%%%%%%%%%%%%%%%%%%%%%%%%%%%%%%%%%%%%%%%%%%%%%%%%%%%%%%
% Fig 1
%%%%%%%%%%%%%%%%%%%%%%%%%%%%%%%%%%%%%%%%%%%%%%%%%%%%%%
\begin{figure}%[!htb]
\centering
\includegraphics[width=10cm]{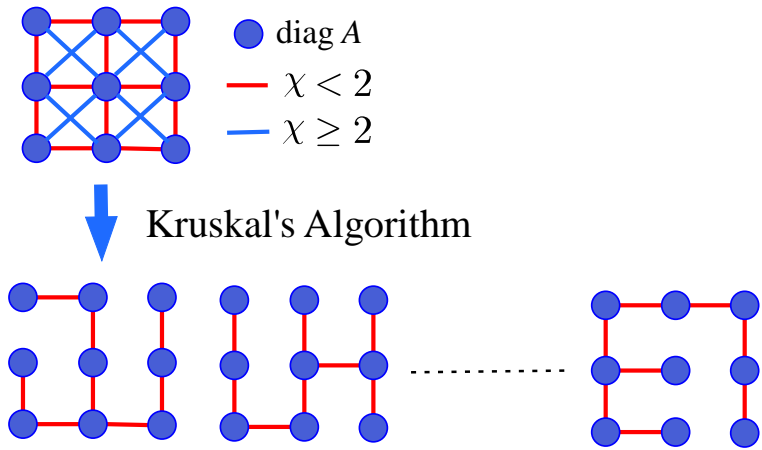}
\caption{Image of Kruskal's Algorithm}
\label{fig:1}
\end{figure}
%%%%%%%%%%%%%%%%%%%%%%%%%%%%%%%%%%%%%%%%%%%%%%%%%%%%%%

Let us focus on $l$-trees over $G$ whose root
is a state $z^l\in \diag(\A)$. 
Recall now that the resistance of the tree
is the sum of the weights of all the edges
constituting the tree as defined in Subsection \ref{sec:2.2}.
Then, we have the following lemma in terms of the 
stochastic potential of $z^l$, which is the minimal
resistance among all $l$-trees in $\G(l)$.
%%%%%%%%%%%%%%%%%%%%%%%%%%%%%%%%%%%%%%%%%%%%%%%%%%%%%%
% Lemma 5
%%%%%%%%%%%%%%%%%%%%%%%%%%%%%%%%%%%%%%%%%%%%%%%%%%%%%%
\begin{lemma}
\label{lem:6}
Consider the weighted directed graph $G$ constituted
from the Markov process $\{P_t^\varepsilon\}$
induced by PHPIP applied to
a constrained potential game $\Gamma$.
Let us denote by $\T = (\diag(\A),\E_l,\W)$ the $l$-tree 
giving the stochastic potential of $z^l \in \diag(\A)$.
If Assumptions \ref{ass:1} and \ref{ass:2} are satisfied, 
then the edge set $\E_l$ must be a subset of $\E_{single}$.
\end{lemma}
%%%%%%%%%%%%%%%%%%%%%%%%%%%%%%%%%%%%%%%%%%%%%%%%%%%%%%
%
\begin{proof}
The edges of $G$, denoted by $\E$, are divided into two classes: 
$\E_s := \E_{single}$ and $\E_d := \E\setminus \E_s$.
From Lemma \ref{lem:3}, the weights of the edges in $\E_s$ 
are smaller than $2$.
We next consider the weights of the edges in $\E_d$.
Because of the nature of PHPIP, any agent cannot change his 
action to another one without explorations when $z(t) \in \diag(\A)$,
and hence exploration should be executed more than twice
in order that the transition along with an edge in $\E_d$ occurs.
This implies that the weights of edges in $\E_d$ should be greater than $2$.

Hereafter, we simply rewrite the weights of the edges $\E_s$
by $w_s (< 2)$ and those of $\E_d$ by 
$w_d (\geq 2)$ and
build the minimal resistance tree with root $z^l$ over 
this simplified graph. 
Note that this simplification does not change 
the elements of the edge set $\E_l$.
It should be noted that from Assumption \ref{ass:1}
all recurrent communication classes ($\diag({\mathcal A})$) can be connected 
by passing only through straight routes.
From the procedure of Kruskal's Algorithm, 
edges with resistances $w_d$ are never 
chosen as edges of the minimal tree as illustrated in Fig. \ref{fig:1}.
Thus, the tree giving the stochastic potential must consist only of
the edges in $\E_s$, which completes the proof.
\end{proof}

We are now ready to state the following proposition
on the stochastically stable states (Definition \ref{def:4}) for 
the Markov process $\{P_t^\varepsilon\}$.
%%%%%%%%%%%%%%%%%%%%%%%%%%%%%%%%%%%%%%%%%%%%%%%%%%%%%%
% Proposition 3
%%%%%%%%%%%%%%%%%%%%%%%%%%%%%%%%%%%%%%%%%%%%%%%%%%%%%%
\begin{proposition}
\label{prop:3}
Consider $\{P_t^\varepsilon\}$ induced by PHPIP applied to
a constrained potential game $\Gamma$. 
If Assumptions \ref{ass:1} and \ref{ass:2} are satisfied, then the 
stochastically stable states are included in 
$\diag(\zeta(\Gamma))$, 
with the set of the optimal Nash equilibria $\zeta(\Gamma)$.
\end{proposition}
%%%%%%%%%%%%%%%%%%%%%%%%%%%%%%%%%%%%%%%%%%%%%%%%%%%%%%
%
\begin{proof}
From Proposition \ref{prop:1}, Lemmas \ref{lem:1} and \ref{lem:2}, it is sufficient
to prove that the states in $\diag (\A)$
with the minimal stochastic potential over $G$
are included in $\zeta(\Gamma)$.

Let us introduce the notations $z_{nonopt} = (a_{nonopt},a_{nonopt})\in \diag (\A)$ 
 with a non optimal action $a_{nonopt}$ 
 and $z_{opt}=(a_{opt},a_{opt})\in \diag({\mathcal A})$ with an
 optimal Nash equilibrium $a_{opt}$.
If $z_{nonopt}$ is the root of a tree $T$, there exists a unique route 
from $z_{opt}$ to $z_{nonopt}$ over $T$.
From Lemma \ref{lem:6}, the route $r$ is an $m$-straight-route for some $m$.
Now, we can build a tree $T'$ with root $z_{opt}$ such that
only the route $r$ is replaced by its reverse route $r'$
(Fig. \ref{fig:2}).
Then, we have $\lambda(r)>\lambda(r')$
from Lemma \ref{lem:4} since $\phi(a_{opt}) > \phi(a_{nonopt})$.
Thus, the resistance of $T'$ is smaller than that of $T$
and the stochastic potential of $z_{opt}$ is smaller than 
the resistance of $T'$.
The statement holds regardless of the selection of $a_{nonopt}$.
%Hence the stochastic potential
%of $z_{nonopt}$ is also larger than that of $z_{opt}$.
This completes the proof.
\end{proof}
%

%%%%%%%%%%%%%%%%%%%%%%%%%%%%%%%%%%%%%%%%%%%%%%%%%%%%%%
% Fig 2
%%%%%%%%%%%%%%%%%%%%%%%%%%%%%%%%%%%%%%%%%%%%%%%%%%%%%%
\begin{figure}%[!htb]
\centering
\includegraphics[width=13cm]{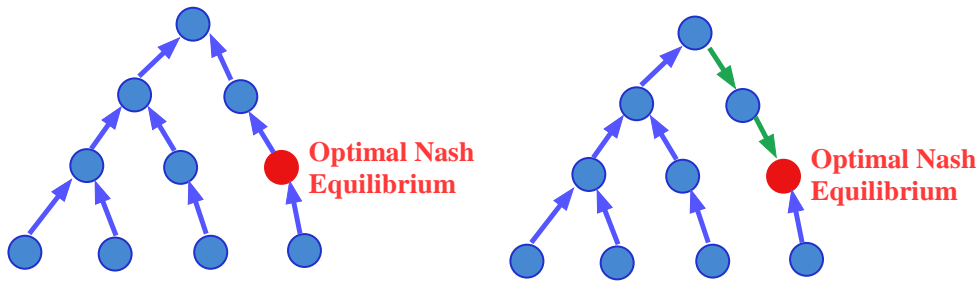}
\caption{Resistance Trees (the left tree should have a greater resistance than the right)}
\label{fig:2}
\end{figure}
%%%%%%%%%%%%%%%%%%%%%%%%%%%%%%%%%%%%%%%%%%%%%%%%%%%%%%

We next consider PIPIP with time-varying $\varepsilon(t)$ and
prove strong ergodicity of $\{P_t^\varepsilon\}$. 
%
%%%%%%%%%%%%%%%%%%%%%%%%%%%%%%%%%%%%%%%%%%%%%%%%%%%%%%
% Lemma 6
%%%%%%%%%%%%%%%%%%%%%%%%%%%%%%%%%%%%%%%%%%%%%%%%%%%%%%
\begin{lemma}\label{lem:8}
The Markov process $\{P_t^\varepsilon\}$ induced by PIPIP 
applied to a constrained potential game $\Gamma$ is strongly ergodic.
\end{lemma}
%%%%%%%%%%%%%%%%%%%%%%%%%%%%%%%%%%%%%%%%%%%%%%%%%%%%%%
%
\begin{proof}
We use Proposition \ref{prop:2} for the proof.
Conditions {\bf (B2)}, {\bf (B3)} in Proposition \ref{prop:2} 
can be proved in the same way as \cite{ZM09a}.
We thus show only the satisfaction of Condition {\bf (B1)}.
As in (\ref{eq:4.1}), the probability of transition 
$z^1\rightarrow z^2$ is given by
\begin{eqnarray}
P_{z^1z^2}^\varepsilon &=& 
\prod_{i\in \Lambda_1} \frac{\varepsilon}{|\R_i(a_i^1)| - 1}
\times \prod_{i\in \Lambda_2} (1 - \varepsilon) \times 
\prod_{i\in \Lambda _3} \frac{\varepsilon}{|\R_i(a_i^1)| - h_i} \times  
\nonumber\\
&&\times \prod_{i\in \Lambda_4} (1 - \varepsilon) \kappa
\varepsilon^{\Delta_i} \times \prod_{i\in \Lambda_5} 
(1 - \varepsilon)(1 - \kappa\varepsilon^{\Delta_i}).
\label{eq:4.11}
\end{eqnarray}
Since $\varepsilon(t)$ is strictly decreasing, there is $t_0\geq 1$ 
 such that $t_0$ is the first time when
\begin{equation}
(1-\varepsilon(t))(1-\kappa\varepsilon(t) ^{\Delta _i })\geq 
 \frac{\varepsilon(t)}{C -1},\
1- \varepsilon(t) \geq \frac{\varepsilon(t) ^{(1-\Delta_i)}}{\kappa(C -1)}
\label{eq:4.12}
\end{equation}
holds.
Note that the existence of $\varepsilon$ satisfying (\ref{eq:4.12}) 
is guaranteed from the condition (\ref{eq:3.4}).
For all $t \geq t_0$, we have
\begin {equation}
P_{z^1z^2}^\varepsilon(t) \geq 
\left(\frac{\varepsilon(t)}{C -1}\right)^n.
\label{eq:4.13}
\end {equation}
The remaining part of the proof is the same as \cite{ZM09a}
and omit it in this paper. 
\end{proof}

We are now ready to prove Theorem \ref{thm:1}.
From Lemma \ref{lem:8}, the distribution $\mu (\varepsilon(t))$ 
converges to the unique distribution $\mu ^*$  from any initial state.
%Moreover, 
%since $\mu ^t$ depends on $\varepsilon(t)$, $\mu ^*$ is given as 
%
%\begin{equation}
%$\mu ^* 
%= {\lim}_{t \rightarrow \infty} \mu ^t
%= {\lim}_{t \rightarrow \infty} \mu(\varepsilon(t))$
%\end{equation}
%
%is true.
In addition, we also have $\mu ^*=\mu (0) = {\lim}_{\varepsilon\rightarrow 0}\mu(\varepsilon)$
from ${\lim}_{t\rightarrow \infty}\varepsilon(t) = 0$.
We have already proved from Propositions \ref{prop:1} and \ref{prop:3}
that any state $z$ satisfying $\mu_z(0) > 0$ must be included 
in $\diag(\zeta(\Gamma))$.
Therefore, 
\[
 {\lim}_{t\rightarrow \infty} {\rm Prob}[z(t) \in \diag(\zeta(\Gamma))] = 1,
\]
is proved, which completes the proof of Theorem \ref{thm:1}.
Theorem \ref{thm:2} is also proved from Proposition \ref{prop:1}, Lemma \ref{lem:1} 
and Proposition \ref{prop:3}.

\section{Application to Sensor Coverage Problem}
\label{sec:5}

In this section we demonstrate the effectiveness of the 
proposed learning algorithm PIPIP through experiments of
the sensor coverage 
problem investigated e.g. in \cite{CMKB04,LC05,CZ08} whose
objective is to cover a mission space efficiently 
using distributed control strategies.
In particular, the problem of this section is formulated based on \cite{ZM09a}
with some modifications.

\subsection{Problem Formulation}
\label{sec:5.1}

We suppose that the mission space to be covered
is given by ${\mathcal Q}^c \subset {\mathbb R}^2$
and that a density function $W^c(q),\ q\in {\mathcal Q}^c$
is defined over ${\mathcal Q}^c$.
In particular, to constitute a game in the form of the previous
sections, 
we also prepare a discretized mission space ${\mathcal Q}$ 
consisting of a finite number of points in ${\mathcal Q}^c$.
%${\mathcal Q}_j,\ j = 1, \cdots, J$ in ${\mathcal Q}^c$
%and let the collection of the centers 
%$q^l = (q_x^{l}, q_y^{l}) \in {\mathcal Q}_l,\ j = 1, \cdots, L$.
Accordingly, we also define the discretized version of the density
$W(q),\ q\in {\mathcal Q}$ such that $W(q) = W^c(q)\ {\forall q}\in {\mathcal Q}$.
%of the squares be denoted by ${\mathcal Q}$. 

In the problem, the position of agent $i$ in the mission space 
${\mathcal Q}$
is regarded as the action $a_i$ to be determined, and hence
the action set $\A_i$ is given by a subset of ${\mathcal Q}$ for all 
$i\in {\V}$. 
Namely, each agent $i$ chooses his action $a_i$ from the finite set 
$\A_i \subseteq {\mathcal Q}$ at each iteration and move toward the corresponding point.

Suppose now that each sensor has a limited sensing radius $r_m$ and
that agent $i$ located at $a_i \in {\mathcal Q}$
may sense an event at $q\in {\mathcal Q}$ iff 
$q \in \D(a_i) := \{q \in {\mathcal Q}|\ \|q - a_i\| \leq r_m\}$.
We also denote by $n_q(a)$ the number of agents 
such that $q \in \D(a_i)$
when agents take the joint action $a$.
Then, we define the function
\begin{equation*}
\phi(a) = \sum_{q\in {\mathcal Q}}\sum^{n_q(a)}_{l=1}\frac{W(q)}{l}dq.
\end{equation*}
This function means, as $n_q(a)$ increases,
the sensing accuracy at 
$q\in {\mathcal Q}$ improves
but the increment decreases, which captures the
characteristics of the sensor coverage problem. 
Note that the authors in \cite{ZM09a} 
take account of energy consumption of sensors
in addition to coverage performance
and claim that the function $\phi$ cannot be a performance measure.
However, we do not consider the energy consumption and what is the best
selection of the performance measure depends on the 
subjective views of designers.
We thus identify maximization of $\phi$
with the global objective of the group letting
$\phi$ be the potential function.

Let us now introduce the utility function 
\begin{equation*}
U_i(a) = \sum_{q\in \D(a_i)}\frac{W(q)}{n_q(a)}.
\end{equation*}
Then, equation (\ref{eq:2.1})
holds for the above potential function $\phi$ \cite{ZM09a}
and hence a potential game 
is constituted.
It is also easy to confirm that the utility $U_i(a)$ can be locally computed 
if we assume feedbacks of $W_q,\ q \in \D(a_i)$ from environment and 
of the selected actions $a_j,\ j\neq i$ 
only from neighboring agents specified by
the $2r_m$-disk proximity communication graph \cite{BCM09}.

\subsection{Objectives}
\label{sec:5.2}

In this section, we run two experiments whose
objectives are listed below.
\begin{itemize}
\item Demonstration of effectiveness: 
Theorems \ref{thm:1} and \ref{thm:2} assure statements after infinitely long time but
it is required in practice that the algorithm works in finite time.
The first objective is thus to 
check if the agents successfully cover the mission space
(i) even in the presence of constraints
such as obstacles and mobility constraints, and 
(ii) in the absence of the prior information on the density function.
The second objective is to compare its performance with
the learning algorithm DISL, which is chosen to ensure
fair comparisons.
Indeed, the other existing algorithms require either or both of
prior knowledge on density or free motion without constraints.
\item Adaptability to environmental changes:
In many real applications of sensor coverage schemes,
it is required for sensors to change the configuration 
according to the surrounding environment.
In particular, the density function can be time-varying
e.g. in the scenario such as measuring of radiation quantity in the air
and sampling of some chemical material and temperature in the ocean. 
It is expected for payoff-based algorithms to naturally 
adapt to such environmental changes
without altering action selection rules and any 
complicated decision-making processes
due to the characteristics that prior knowledge on
environments except for $\A_i$ is not assumed.
We thus check the function by using a Gaussian density function whose
mean moves as time advances.
\end{itemize}

\subsection{Experimental System}
\label{sec:5.3}

%%%%%%%%%%%%%%%%%%%%%%%%%%%%%%%%%%%%%%%%%%%%%%%%%%%%%%
% Fig 3
%%%%%%%%%%%%%%%%%%%%%%%%%%%%%%%%%%%%%%%%%%%%%%%%%%%%%%
\begin{figure}
\begin{center}
\includegraphics[width=10cm]{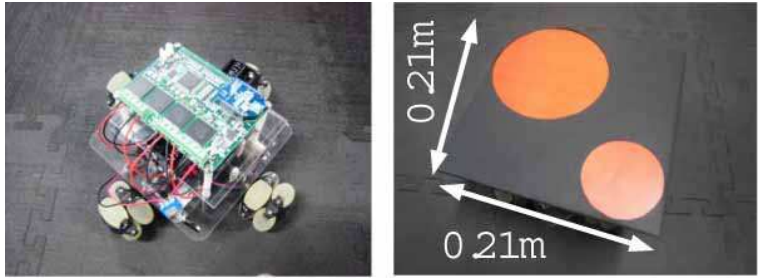}
\caption{Mobile Robot}
\label{fig:3}
\end{center}
\end{figure}
%%%%%%%%%%%%%%%%%%%%%%%%%%%%%%%%%%%%%%%%%%%%%%%%%%%%%%

In the experiments, we use four mobile robots with
four wheels which can move in any direction (Fig. \ref{fig:3}).
Fig. \ref{schematic} shows the schematic of the experimental system.
A camera (Firefly MV (ViewPLUS Inc.) with
lenses LTV2Z3314CS-IR (Raymax Inc.)) is mounted over the field.
The image information is sent to a PC and processed
to extract the pose of robots from the image
by the image processing library OpenCV 2.0.
Note that a board with two colored feature points is attached to each 
robot as in Fig. \ref{fig:3} to help the extraction.
According to the extracted poses, the actions to be taken by agents
are computed based on learning algorithms.
However, in the experiments, the selected actions are not executed directly
since collisions among robots must be avoided.
For this purpose, a local decision-making mechanism 
checks whether collisions would occur
if the selected actions were executed.
The mechanism is designed based on heuristics and
we avoid mentioning the details since 
it is not essential.
If the answer of the mechanism is yes, the agents
decide to stay at the current position.
Otherwise, the selected actions are sent 
as reference positions together with 
the current poses to the local
velocity and position PI controller implemented
on a digital signal processor DS-1104
(dSPACE Inc.).
Then, the eventual velocity command is sent to each robot
via a wireless communication device XBee
(Digi International Inc.).

%%%%%%%%%%%%%%%%%%%%%%%%%%%%%%%%%%%%%%%%%%%%%%%%%%%%%%
% Fig 4, 5
%%%%%%%%%%%%%%%%%%%%%%%%%%%%%%%%%%%%%%%%%%%%%%%%%%%%%%
\begin{figure}
\begin{center}
\begin{minipage}{10cm}
\includegraphics[width=10cm]{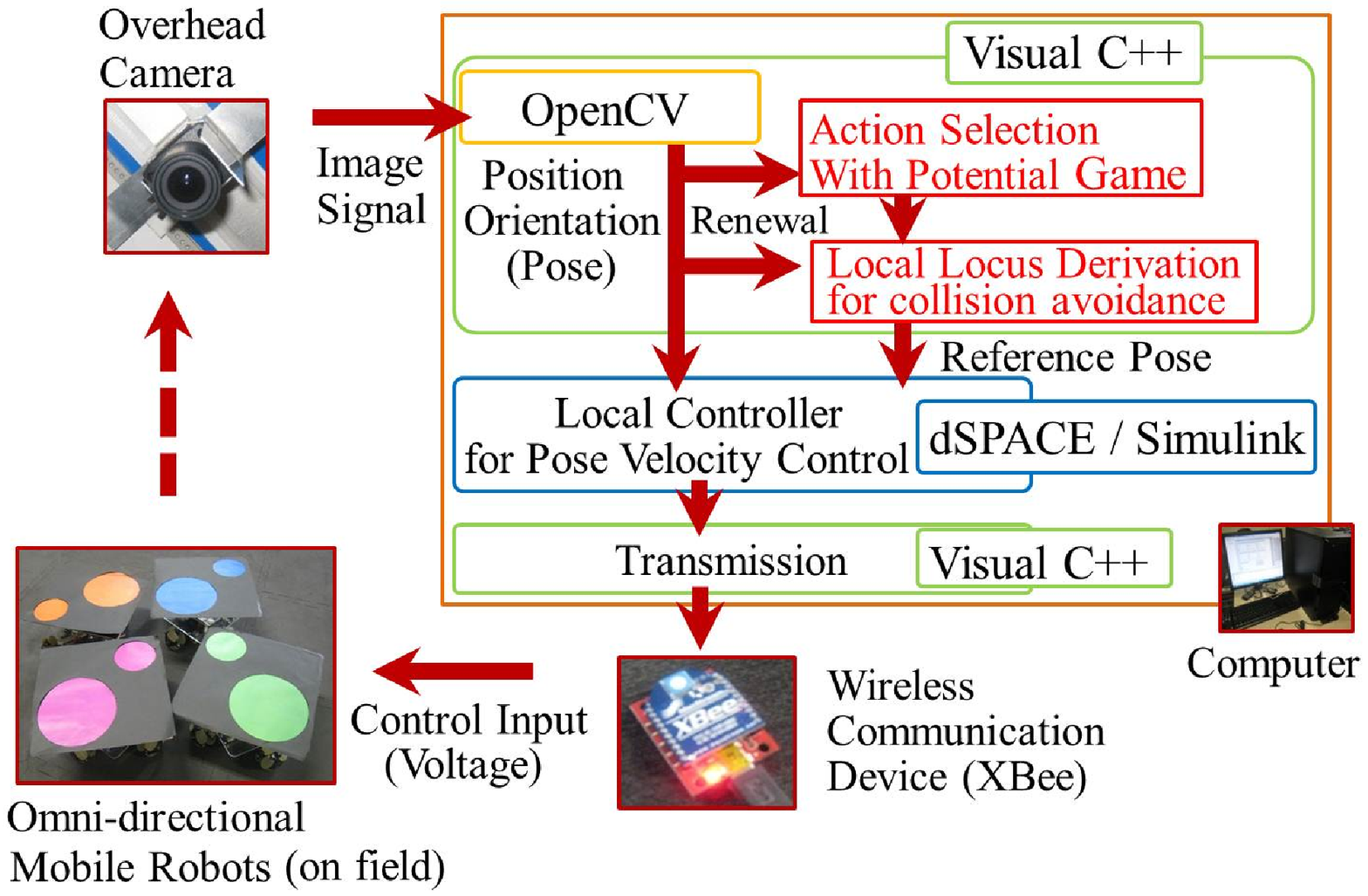}
\caption{Experimental Schematic}
\label{schematic}
\end{minipage}
\begin{minipage}{6cm}
\includegraphics[width=6cm]{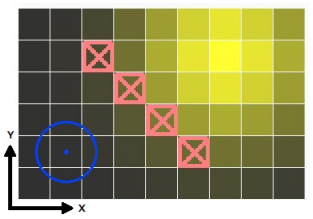}
\caption{Setting of Experiment 1}
\label{field}
\end{minipage}
\end{center}
\end{figure}
%%%%%%%%%%%%%%%%%%%%%%%%%%%%%%%%%%%%%%%%%%%%%%%%%%%%%%

The following setup is common in all experiments.
The mission space ${\mathcal Q}^c := [0\ 2.7]{\rm m}\times [0\ 1.8]{\rm 
m}$ 
is divided into $9 \times 6$ squares with side length $0.3$m
as in Fig. \ref{field}
letting the discretized set ${\mathcal Q}$ be given by the
centers of the squares as
%the action sets are given by
\[
 {\mathcal Q} = \{(0.15+0.3j, 0.15+0.3l)|\ j \in \{0, \cdots, 8\},\ l\in \{0, 
\cdots, 5\}\}.
\]
The sensing radius $r_m$ is set as $r_m = 0.3$m for all robots.
We also assume that each agent has a mobility constraint
\[
 \R_i(a_i) = \{a_i \pm 0.3 (b_1,b_2)\in \A_i|\ b_1 \in \{ -1, 0 , 1\},\ 
b_2 \in \{-1, 0, 1\}\}.
\]
The initial actions of agents are set as
\[
 a_1(0) = (0.15,0.15),\ a_2(0) = (0.15, 0.45),\ a_3(0) = (0.45,0.15),\ 
a_4(0) = (0.45, 0.45).
\]

\subsection{Experiment 1}
\label{sec:5.4}

%%%%%%%%%%%%%%%%%%%%%%%%%%%%%%%%%%%%%%%%%%%%%%%%%%%%%%
% Fig 6
%%%%%%%%%%%%%%%%%%%%%%%%%%%%%%%%%%%%%%%%%%%%%%%%%%%%%%
\begin{figure}
\begin{center}
\includegraphics[width=15cm]{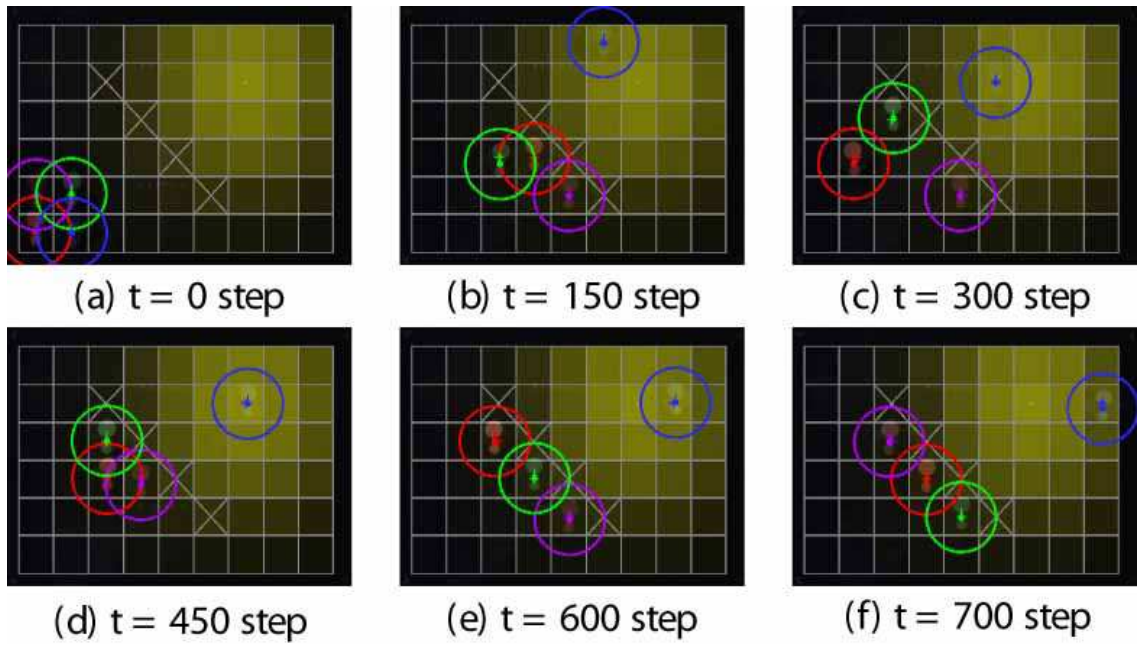}
\caption{Configurations by DISL (Experiment 1)}
\label{obs:disl}
\end{center}
\end{figure}
%%%%%%%%%%%%%%%%%%%%%%%%%%%%%%%%%%%%%%%%%%%%%%%%%%%%%%

In the first experiment, we demonstrate the effectiveness of PIPIP.
For this purpose, we employ the density function 
\begin{eqnarray}
 W(q) = e^{-\frac{25\|q - \mu\|^2}{9}},\ \mu = (1.95,1.35)
\nonumber
\end{eqnarray}
and prepare obstacles at
\begin{eqnarray}
{\mathcal O} := \{(0.75, 1.35), (1.05,1.05), (1.35, 0.75), (1.65,0.45)\}.
\end{eqnarray}
Namely, in the experiment, the action sets are given by 
${\A}_i = {\mathcal Q}\setminus {\mathcal O}$.
The setup is illustrated in Fig. \ref{field}, 
where the region with high density is colored by yellow
and the red cross mark indicates the actions prohibited to be taken by the obstacles.
Under the situation, we see that there exist some Nash undesirable 
equilibria just ahead on the left of the obstacles.
It should be also noted that each robot does not know the function $W(q)$
{\it a priori}.

We first run DISL under the above situation 
with the exploration rate $\varepsilon= 0.15$.
Then, the resulting configurations 
at $0$, $150$, $300$, $450$, $600$ and $700$ steps
are shown in Fig. \ref{obs:disl}.
Under the setting, three robots cannot 
reach the colored region at least in $700$ step.
It is now easily confirmed that the configurations at $600$ and $700$[step]
are Nash equilibria only for the three robots and hence they
cannot increase utilities by any one agent's action change.
%%%%%%%%%%%%%%%%%%%%%%%%%%%%%%%%%%%%%%%%%%%%%%%%%%%%%%
% Fig 7
%%%%%%%%%%%%%%%%%%%%%%%%%%%%%%%%%%%%%%%%%%%%%%%%%%%%%%
\begin{figure}
\begin{center}
\includegraphics[width=15cm]{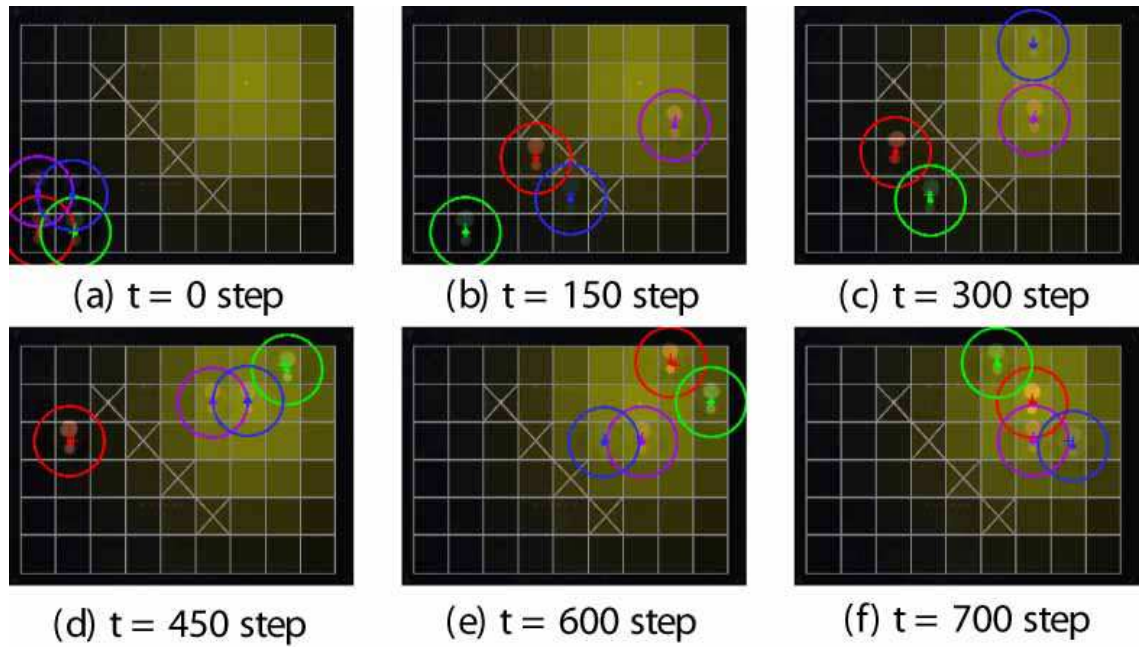}
\caption{Configurations by PIPIP (Experiment 1)}
\label{obs:pipip}
\end{center}
\end{figure}
%%%%%%%%%%%%%%%%%%%%%%%%%%%%%%%%%%%%%%%%%%%%%%%%%%%%%%

We next run PIPIP
letting the parameter $\varepsilon$
be fixed as $\varepsilon= 0.15$ and setting $\kappa = 0.5$
(namely, PHPIP is actually run in the experiment).
Fig. \ref{obs:pipip} shows resulting configurations
at the same steps as Fig. \ref{obs:disl}.
Surprisingly, we see that all the robots eventually avoid
the obstacles and arrive at the colored region though
they initially do not know where is important.
Such a behavior is never achieved by conventional 
coverage control schemes.
%We also see that they make no attempt to leave the colored region 
%after reaching there.
The time responses of the potential function $\phi$ 
for PIPIP and DISL are illustrated in Fig. \ref{obs:pot},
where the solid line shows the response for PIPIP and
the dashed line for DISL.
As is apparent from the above investigations, PIPIP achieves a higher potential
function value than DISL.
%In addition, we also see from these results that
%the effect of space discretization
%is not significant in the setup and PIPIP achieves almost 
%desirable configurations even in the continuous space.

Though we can show only one sample due to the page constraints,
similar results are obtained for both DISL and PIPIP through several trials.
From the results, we claim that PIPIP has a stronger tendency to
escape undesirable Nash equilibria than DISL, which is also confirmed by 
the meaning of the irrational decision.
Of course, the results strongly depend on the value of exploration 
rate $\varepsilon$.
We thus show the time evolution of the function $\phi$
for $\varepsilon= 0.3$ in Fig. \ref{obs:pot1}.
We see from Fig. \ref{obs:pot1} that some agents executing DISL 
also do not reach the important region even for 
$\varepsilon= 0.3$, which seems to be quite high probability
as an exploration rate.
Indeed, the fluctuation of the responses is large and
an agent with PIPIP overcomes the obstacle again 
leaving the colored region.
From all the above results, we thus can state 
that guarantees of only convergence to Nash equilibria
can be a significant problem not only
from the theoretical point of view but also from 
the practical viewpoint.
Though much more thorough comparisons are necessary
in order to make the claim on superiority of PIPIP over DISL confident,
PIPIP achieves a better performance than DISL at least in the setup.

%%%%%%%%%%%%%%%%%%%%%%%%%%%%%%%%%%%%%%%%%%%%%%%%%%%%%%
% Fig 8, 9
%%%%%%%%%%%%%%%%%%%%%%%%%%%%%%%%%%%%%%%%%%%%%%%%%%%%%%
\begin{figure}
\begin{center}
\begin{minipage}{7.8cm}
\includegraphics[width=7.8cm]{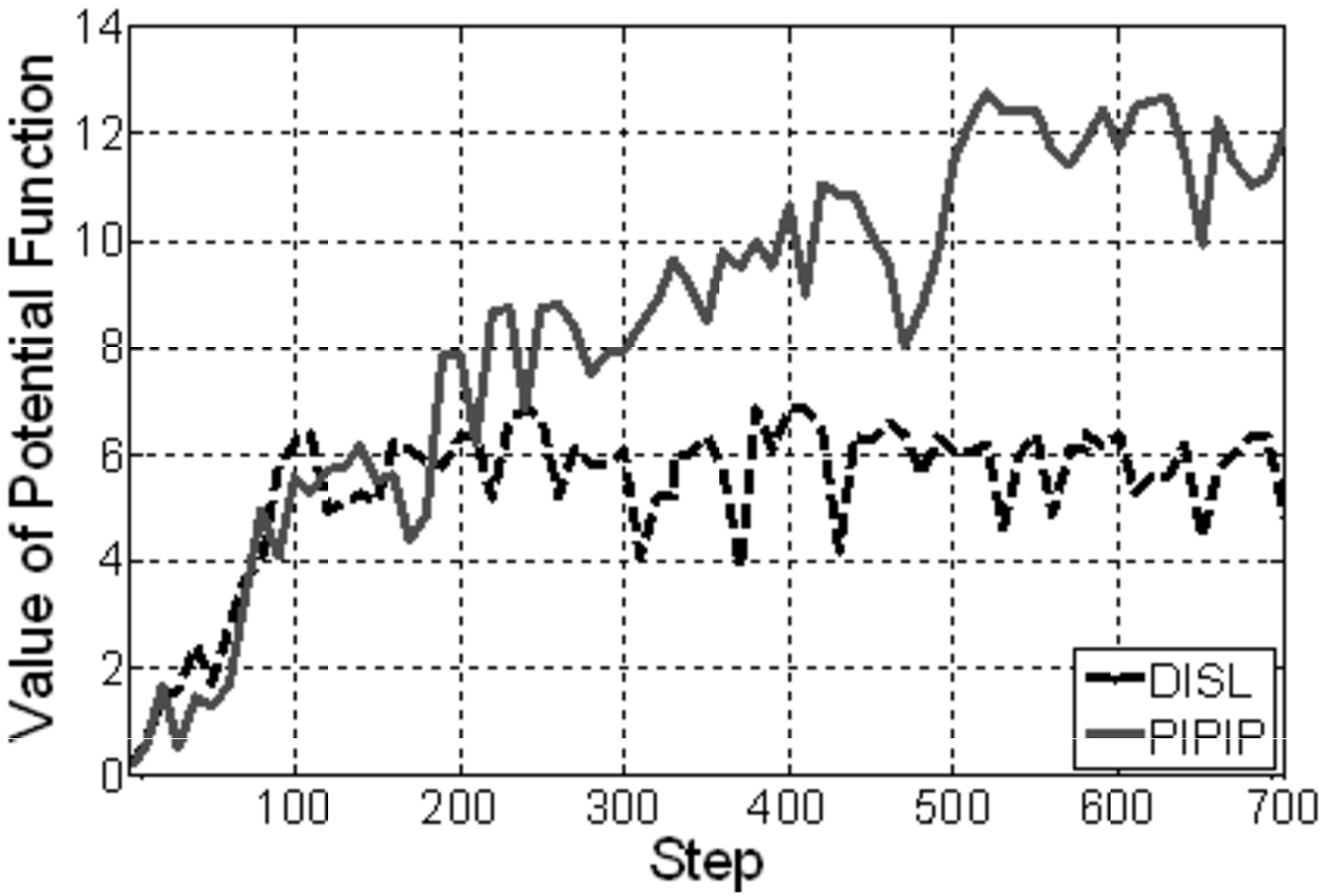}
\caption{Time Evolution of Potential Function for $\varepsilon= 0.15$ (Experiment 1)}
\label{obs:pot}
\end{minipage}
\hspace{.4cm}
\begin{minipage}{7.8cm}
\includegraphics[width=7.8cm]{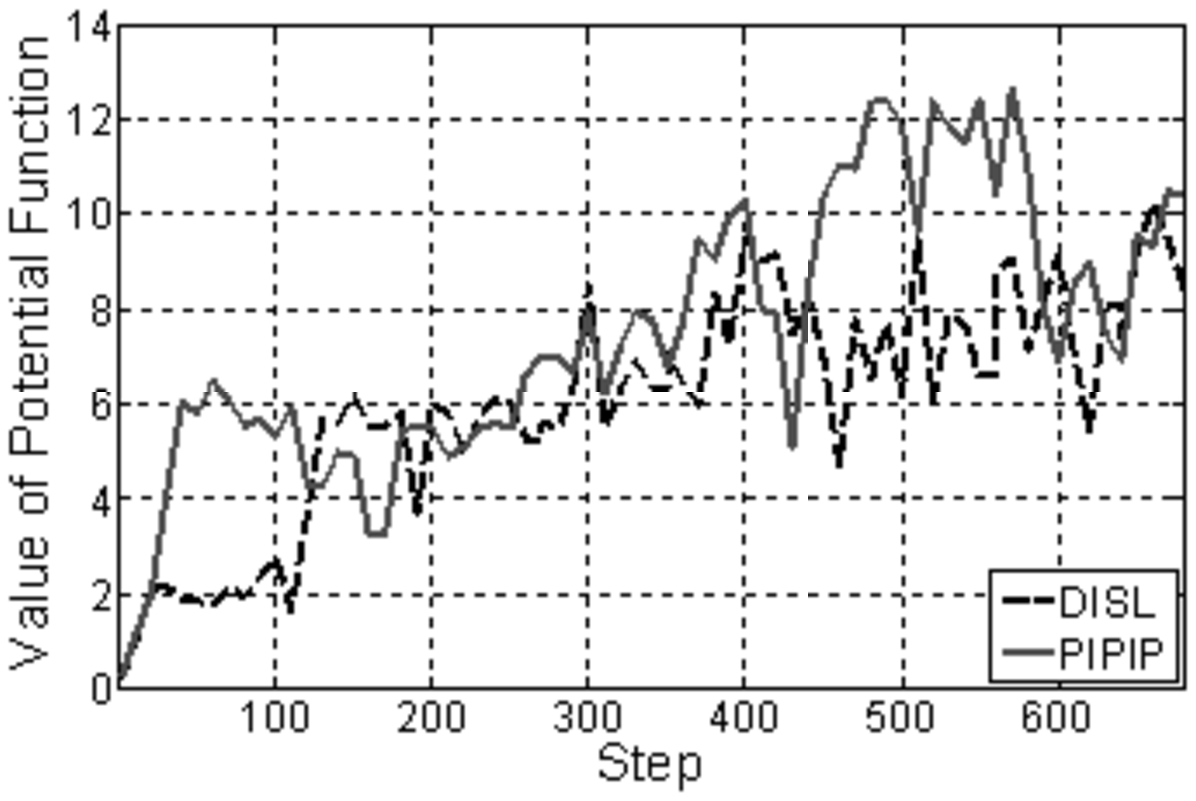}
\caption{Time Evolution of Potential Function for $\varepsilon= 0.3$ (Experiment 1)}
\label{obs:pot1}
\end{minipage}
\end{center}
\end{figure}
%%%%%%%%%%%%%%%%%%%%%%%%%%%%%%%%%%%%%%%%%%%%%%%%%%%%%%

\subsection{Experiment 2}
\label{sec:5.5}

We next demonstrate the adaptability of PIPIP
to environmental changes, where
we get rid of the obstacle ${\mathcal O}$
and hence $\A_i = {\mathcal Q}$.
In the experiment, we use the following Gaussian density function
whose mean gradually moves.
\begin{eqnarray}
 W^c(q) = e^{-\frac{25\|q - \mu(t)\|^2}{9}},\ \mu(t) = 
\left\{
\begin{array}{ll}
(0.45,0.45),&\mbox{if }t \in [0, 300]\\
(0.00375t-0.6750,0.00225t-0.225),&\mbox{if } t\in 
 (300,700)\\
(1.95,1.35),&\mbox{if }t\geq 700
\end{array}
\right.
\nonumber
\end{eqnarray}
It is worth noting that agents select actions without using
any prior information on the density.

%%%%%%%%%%%%%%%%%%%%%%%%%%%%%%%%%%%%%%%%%%%%%%%%%%%%%%
% Fig 10, 11
%%%%%%%%%%%%%%%%%%%%%%%%%%%%%%%%%%%%%%%%%%%%%%%%%%%%%%
\begin{figure}
\begin{center}
\includegraphics[width=15cm]{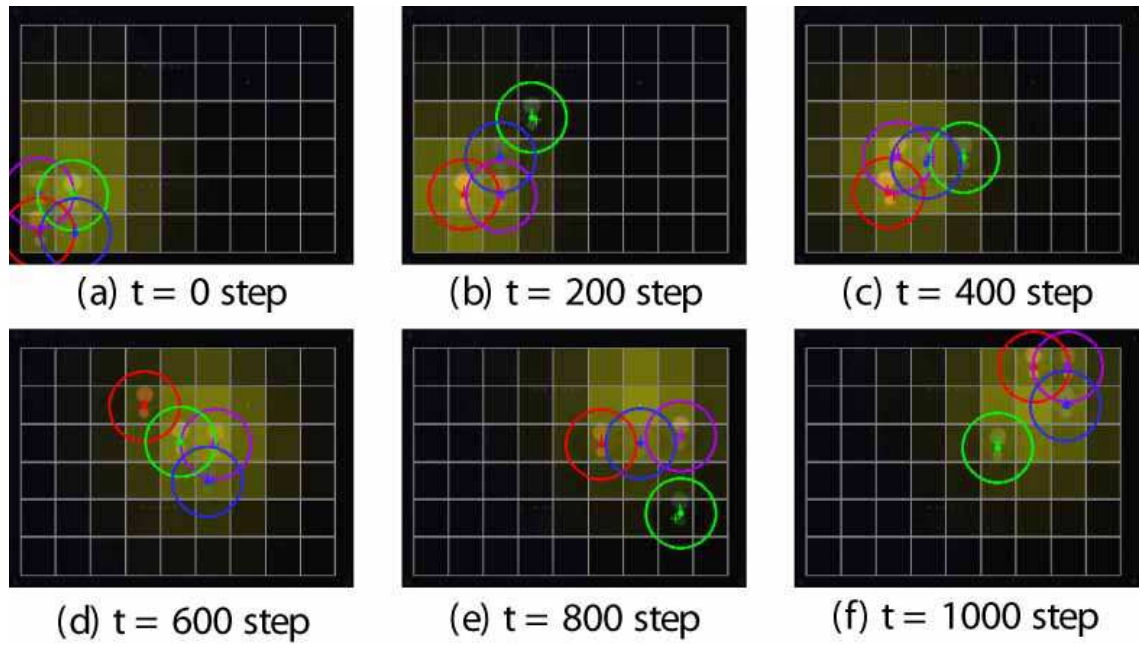}
\caption{Configurations by PIPIP (Experiment 2)}
\label{mov:pipip}
\end{center}
\end{figure}
\begin{figure}
\begin{center}
\includegraphics[width=8cm]{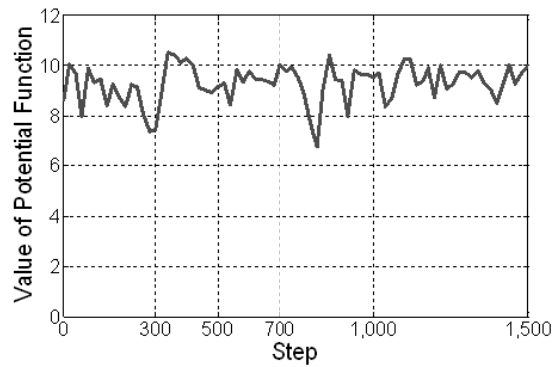}
\caption{Time Evolution of Potential Function (Experiment 2)}
\label{mov:pot}
\end{center}
\end{figure}
%%%%%%%%%%%%%%%%%%%%%%%%%%%%%%%%%%%%%%%%%%%%%%%%%%%%%%

Figs \ref{mov:pipip} and \ref{mov:pot} respectively illustrate the
resulting configurations at $0$, $200$, $400$, $600$, $800$
and $1000$ steps and time evolution of the potential function $\phi$.
We see from Fig. \ref{mov:pipip} that agents 
gather at around the most 
important region at any time instant while
learning the environmental changes.
Fig. \ref{mov:pot} also shows that the potential function
keeps almost the same level during whole time, which indicates that
the agents successfully track the most important region.
From these results, as expected, agents executing 
PIPIP successfully adapt to the environmental
changes without changing the action selection rule at all.
Such a behavior is also never achieved by conventional 
coverage control schemes.

\section{Conclusion}
\label{sec:6}

In this paper, we have developed a new learning algorithm Payoff-based
Inhomogeneous Partially Irrational Play (PIPIP) for potential game theoretic
cooperative control of multi-agent systems.
The present algorithm is based on Distributed Inhomogeneous 
Synchronous Learning (DISL) presented in \cite{ZM09a} and inherits
several desirable features of DISL.
However, unlike DISL, PIPIP
allows agents to make irrational decisions, that is,
take an action giving a lower utility from the past two actions.
Thanks to the decision, we have succeeded proving convergence of the joint action
to the potential function maximizers while
escaping from undesirable Nash equilibria.
Then, we have demonstrated the utility of PIPIP through experiments on
a sensor coverage problem. 
It has been revealed through the demonstration that
the present learning algorithm  works even in a finite-time
interval and agents successfully arrive at around the optimal Nash equilibria 
in the presence of obstacles in the mission space.
In addition, we also have seen through an experiment
with a moving density function that PIPIP has adaptability
to environmental changes, which is a function expected for payoff-based
learning algorithms.

%~\cleardoublepage
%------------------------------------------------------------------------%
% References
%------------------------------------------------------------------------%

%~\cleardoublepage~\newpage

%------------------------------------------------------------------------%
% Appendix
%------------------------------------------------------------------------%

% \appendix
% \input{appendix4}


\begin{thebibliography}{99}
%%%%%%%%%%%%%%%%%%%%%%%%%%%%%%%%%%%%%%%%%%%%%%%%%%%%%%%%%%%%%%%%%%%%%%%%%%%%%%%%%%%%%%%%%%%%%%%%%%%%%%%%%%%%%%%%%%%%%%%%%%%%%%%%%%%%%

\bibitem{BCM09}
F. Bullo,~J. Cortes and S. Martinez,
{\it Distributed Control of Robotic Networks},
Series in Applied Mathematics, Princeton University Press, 2009.

\bibitem{M07}
R. M. Murray,
``Recent Research in Cooperative Control of Multivehicle Systems,''
{\it Journal of Dynamic Systems Measurement and Control-Transactions of The 
 Asme}, Vol. 129, No. 5, pp. 571--583, 2007.
%\bibitem{V93}
%P. Varaiya, ``Smart Cars on Smart Roads: Problems of Control,'' 
%{\it IEEE Transactions on Automatic Control}, Vol. 38, No. 2, pp. 195--207, 1993.
%\bibitem{SFM07}
%R. Olfati-Saber, J. A. Fax and R. M. Murray, 
%``Consensus and Cooperation in Networked Multi-Agent Systems,'' 
%{\it Proc. of the IEEE}, Vol.
%95, No. 1, 2007.
%
%\bibitem{PLSGP07}
%D. A. Paley, N. E. Leonard, R. Sepulchre, D. Grunbaum
%and J. K. Parrish:
%Oscillator Models and Collective Motion;{\it IEEE Control Systems Magazine},
% Vol 27, No 4, pp. 89-105 (2007)
%\bibitem{LPL07} 
%N. E. Leonard, D. A. Paley, F. Lekien, R. Sepulchre, D. M. Fratantoni 
%and R. E. Davis:
%Collective Motion, Sensor Networks, and Ocean
%Sampling;
%{\it Proc. of the IEEE}, Vol. 95, No. 1, pp. 48-74 (2007)
%\bibitem{CBM07}
%S. Martinez, J. Cortes and F. Bullo:
%Motion Coordination with Distributed Information;
%{\it IEEE Control System Magazine}, Vol. 27, No. 4,
%pp. 75--88 (2007)
\bibitem{CMKB04}
J. Cortes, S.Martinez, T. Karatas and F. Bullo, 
``Coverage Control for Mobile Sensing Networks,''
{\it IEEE Trans. on Robotics and Automation}, Vol. 20, 
No. 2, pp. 243--255, 2004.

\bibitem{LC05}
W. Li and C. G. Cassandras, 
``Sensor Networks and Cooperative Control,'' 
{\it European Journal of Control}, Vol. 11, pp. 436--463, 2005.

\bibitem{CZ08}
C. H. Caicedo-N and M. Zefran,
``A Coverage Algorithm for A Class of Non-convex Regions,''
in {\it Proc. of the 47th IEEE International Conference on Decision and 
	Control}, pp. 4244--4249, 2008.

\bibitem{MAS09}
J. R. Marden, G. Arslan and J. S. Shamma,
``Cooperative Control and Potential Games,''
{\it IEEE  Trans. on Systems,  Man and Cybernetics}, Vol. 39, No. 6, 
	pp. 1393--1407, 2009.

\bibitem{ZM09a}
M. Zhu and S. Martinez, 
``Distributed Coverage Games for Mobile Visual Sensor Networks,''
{\it SIAM Journal on Control and Optimization}, submitted 
(downloadable at arXiv:1002.0367v1), 2010.
%M. Zhu and S. Martinez, 
%``Distributed Coverage Games for Mobile
%Visual Sensors (i): Reaching the Set of Nash Equilibria,'' 
%Proc. of the 48th IEEE Conference. on Decision and Control and 28th Chinese Control
%Conference, pp. 169--174, 2009.

%\bibitem{ZM09b}
%M. Zhu and S. Martinez,
%``Distributed Coverage Games for Mobile
%Visual Sensors (ii): Reaching the Set of Nash Equilibria,'' 
%Proc. of the 48th IEEE Conference. on Decision and Control and 28th Chinese Control
%Conference, pp. 175--180, 2009.

\bibitem{LM10}
N. Li and J. R. Marden, 
``Designing Games to Handle Coupled Constraints'', 
in {\it Proc. of the 49th IEEE Conference on Decision and Control}, pp. 250--255, 2010. 

\bibitem{MS96}
D. Monderer and L. Shapley,
``Potential Games,''
{\it Games and Economic Behavior},
Vol. 14, No. 1, pp. 124--143, 1996.

\bibitem{wierman}
R. Gopalakrishnan, J. R. Marden and A. Wierman,
``An Architectural View of Game Theoretic Control,''
in {\it Proc. of ACM Hotmetrics 2010: 
Third Workshop on Hot Topics in Measurement and Modeling of 
	Computer Systems}, 2010.

\bibitem{S_53}
L. S. Shapley,
{\it A Value for $n$-person Games},
Contributions to the Theory of Games II, 
Princeton University Press, 1953.

\bibitem{WT_99}
D. Wolpert and K. Tumor, 
{\it An Overview of Collective Intelligence},
J. M. Bradshaw Eds.
Handbook of Agent Technology, AAAI Press/MIT Press, 1999.

\bibitem{MS96b}%fictious play
D. Monderer and L. Shapley,
``Fictitious Play Property for Games with Identical Interests,''
{\it Journal of Economic Theory},
Vol. 68, pp. 258--265 1996.

\bibitem{SA03}%regret matching
S. Hart and A. Mas-Colell,
``Regret-based Continuous-time Dynamics,'' 
{\it Games and Economic Behavior},
Vol. 45, No. 2, pp. 375--394, 2003.

\bibitem{JGS09}
J. R. Marden, G. Arslan and J. S. Shamma,
``Joint Strategy Fictitious Play with Inertia for Potential Games,''
{\it IEEE Trans. on Automatic Control}, Vol. 54, No. 2,  
pp. 208--220, 2009.
\bibitem{JGS07}
J. R. Marden, G. Arslan and J. S. Shamma,
``Regret Based Dynamics: Convergence in Weakly Acyclic Games,'' 
in {\it Proc. of 
Sixth International Joint Conference on Autonomous Agents and Multi-Agent Systems}, 2007.

\bibitem{Y93}%adaptive play
H. P. Young,
``The Evolution of Conventions,''
{\it Econometrica},
Vol. 61, No. 1, pp. 57--84, 1993.

\bibitem{Y04}%Better reply process with finite memory and inertia
H. P. Young,
{\it Strategic Learning and Its Limits},
Oxford University Press, 2004.

\bibitem{Young}%sap
H. P. Young,
{\it Individual Strategy and Social Structure: 
An Evolutionary Theory of Institutions},
Princeton University Press, 2001.

\bibitem{JYGS09}
J. R. Marden, H. P. Young, G. Arslan and J. S. Shamma,
``Payoff-based Dynamics for Multi-player Weakly Acyclic Games,'' 
{\it SIAM Journal on Control and Optimization}, Vol. 48, No. 1, pp. 373--396, 2009.

\bibitem{MS08} %payoff based version of log-linear learning
J. R. Marden and J. S. Shamma,
``Revisiting Log-linear Learning:
Asynchrony, Completeness and Payoff-based Implementation,'' 
{\it Games and Economic Behavior}, submitted, 2008.


% \bibitem{HIGF09}
% T. Hatanaka, T. Ibuki, A. Gusrialdi and M. Fujita:
% Coverage Control for Camera Sensor Networks: Its Implementation and Experimental Verification;
% {\it Proc. of the 17th Mediterranean Conference on Control and Automation}, pp. 446-451 (2009)



\bibitem{CS09} 
G. Chasparis and J. Shamma, 
``Distributed Dynamic Reinforcement of Efficient Outcomes in Multiagent Coordination,''
in {\it Proc. of Third World Congress of the Game Theory Society}, 2008. 

%\bibitem{GHF10b}
%T. Goto, T. Hatanaka and M. Fujita, 
%``Potential Game Theoretic Attitude Coordination on the Circle: Synchronization and Balanced Circular Formation,''
%Proc. of the 2010 IEEE Multi-Conference on Systems and Control, pp. 2314--2319, 2010.

\bibitem{er}
D. Isaacson and R. Madsen, 
{\it Markov Chains: Theory and Applications}, New York, Wiley, 1976.

\normalsize
%%%%%%%%%%%%%%%%%%%%%%%%%%%%%%%%%%%%%%%%%%%%%%%%%%%%%%%%%%%%%%%%%%%%%%%%%%%%%%%%%%%%%%%%%%%%%%%%%%%%%%%%%%
\end{thebibliography}
\end{document}